\documentclass[]{llncs}

\usepackage{lmodern}
\usepackage[T1]{fontenc}
\usepackage[utf8]{inputenc}
\usepackage{microtype}
\input{glyphtounicode}
\usepackage{amsmath}
\usepackage{amsfonts}
\usepackage{amssymb}

\usepackage{mathtools}
\mathtoolsset{mathic=true}

\usepackage{tikz}
\usetikzlibrary{shapes,positioning,fit,fadings,calc,decorations.pathmorphing}
\pgfdeclarelayer{background}
\pgfdeclarelayer{lowerBG}
\pgfsetlayers{lowerBG,background,main}

\usepackage[ruled, vlined, linesnumbered]{algorithm2e}
\SetKw{KwDownTo}{downto}

\usepackage{lineno}
\modulolinenumbers[5]

\usepackage{graphicx}
\usepackage{enumerate}
\usepackage{tabto}
\usepackage{xspace}

\usepackage[caption=false]{subfig}

\newcommand{\calA}{\mathcal{A}}
\newcommand{\calB}{\mathcal{B}}
\newcommand{\calC}{\mathcal{C}}
\newcommand{\calD}{\mathcal{D}}

\newcommand{\calL}{\mathcal{L}}

\newcommand{\calO}{\mathcal{O}}
\newcommand{\calP}{\mathcal{P}}

\newcommand{\calS}{\mathcal{S}}
\newcommand{\calT}{\mathcal{T}}

\newcommand{\calV}{\mathcal{V}}

\newcommand{\bbB}{\mathbb{B}}

\newcommand{\bbL}{\mathbb{L}}

\newcommand{\bbN}{\mathbb{N}}

\newcommand{\bbP}{\mathbb{P}}

\newcommand{\bbT}{\mathbb{T}}

\title
{%
    Parameterized Approximation Algorithms for some Location Problems in Graphs
}

\author
{%
    Arne Leitert
    \and
    Feodor F. Dragan
}

\institute
{%
    Department of Computer Science, \\
    Kent State University, Kent, Ohio, USA \\
    \email{aleitert@cs.kent.edu},
    \email{dragan@cs.kent.edu}
}

\makeatletter
\newcommand{\ie}{i.\,e.\@ifnextchar{,}{}{~}}
\newcommand{\eg}{e.\,g.\@ifnextchar{,}{}{~}}
\makeatother

\DeclareRobustCommand{\rHrt}
{%
    \ifmmode
        \text{\boldmath (\raisebox{-0.25ex}{$\heartsuit$})}~%
    \else%
        {\boldmath (\raisebox{-0.25ex}{$\heartsuit$})}\xspace%
    \fi%
}

\DeclareRobustCommand{\rDmd}
{%
    \ifmmode%
        \text{\boldmath ($\diamondsuit$)}~%
    \else%
        {\boldmath ($\diamondsuit$)}\xspace%
    \fi%
}

\newcommand{\Sup}{S^{\raisebox{0.1ex}{$\scriptscriptstyle \uparrow$}}}
\newcommand{\Sdown}{S^{\scriptscriptstyle\downarrow}}
\newcommand{\calSdown}{\calS^{\scriptscriptstyle\downarrow}}

\DeclareMathOperator{\ecc}{ecc}
\DeclareMathOperator{\diam}{diam}
\DeclareMathOperator{\tb}{tb}
\DeclareMathOperator{\tl}{tl}

\DeclareRobustCommand{\qedClaim}
{%
  \ifmmode \lozenge%
  \else%
    \leavevmode\unskip\penalty9999 \hbox{}\nobreak\hfill%
    \quad\hbox{$\lozenge$}%
  \fi%
}%

\DeclareRobustCommand{\qed}
{%
  \ifmmode \square%
  \else%
    \leavevmode\unskip\penalty9999 \hbox{}\nobreak\hfill%
    \quad\hbox{$\square$}%
  \fi%
}%

\spnewtheorem*{conj}{Conjecture}{\normalfont\bfseries}{\itshape}

\usepackage{mathtools}

\begin{document}
\pagestyle{plain}
\maketitle

% \linenumbers

\begin{abstract}
We develop efficient parameterized, with additive error, approximation algorithms for the (Connected) $r$-Domination problem and the (Connected) $p$-Center problem for unweighted and undirected graphs.
Given a graph~$G$, we show how to construct a (connected) $\big(r + \calO(\mu) \big)$-dominating set~$D$ with $|D| \leq |D^*|$ efficiently.
Here, $D^*$ is a minimum (connected) $r$-dominating set of~$G$ and $\mu$ is our graph parameter, which is the \emph{tree-breadth} or the \emph{cluster diameter in a layering partition} of~$G$.
Additionally, we show that a $+ \calO(\mu)$-approximation for the (Connected) $p$-Center problem on~$G$ can be computed in polynomial time.
Our interest in these parameters stems from the fact that in many real-world networks, including Internet application networks, web networks, collaboration networks, social networks, biological networks, and others, and in many structured classes of graphs these parameters are small constants.
\end{abstract}

\section{Introduction}
The (Connected) $r$-Domination problem and the (Connected) $p$-Center problem, along with the $p$-Median problem, are among basic facility location problems with many applications in data clustering, network design, operations research~-- to name a few.
Let $G = (V, E)$ be an unweighted and undirected graph.
Given a radius~$r(v) \in \bbN$ for each vertex~$v$ of~$G$, indicating within what radius a vertex~$v$ wants to be served, the \emph{\( r \)-Domination problem} asks to find a set $D \subseteq V$ of minimum cardinality such that $d_G(v, D) \leq r(v)$ for every $v \in V$.
The \emph{Connected \( r \)-Domination problem} asks to find an $r$-dominating set~$D$ of minimum cardinality with an additional requirement that $D$ needs to induce a connected subgraph of~$G$.
When $r(v) = 1$ for every $v \in V$, one gets the classical (Connected) Domination problem.
Note that the Connected $r$-Domination problem is a natural generalization of the Steiner Tree problem (where each vertex~$t$ in the target set has $r(t) = 0$ and each other vertex~$s$ has $r(s) = \diam(G)$).
The connectedness of~$D$ is important also in network design and analysis applications (\eg in finding a small backbone of a network).
It is easy to see also that finding minimum connected dominating sets is equivalent to finding spanning trees with the maximum possible number of leaves.

The (closely related) \emph{\( p \)-Center problem} asks to find in~$G$ a set~$C \subseteq V$ of at most $p$~vertices such that the value $\max_{v \in V} d_G(v, C)$ is minimized.
If, additionally, $C$ is required to induce a connected subgraph of~$G$, then one gets the \emph{Connected \( p \)-Center problem}.

The domination problem is one of the most well-studied NP-hard problems in algorithmic graph theory.
To cope with the intractability of this problem it has been studied both in terms of approximability (relaxing the optimality) and fixed-parameter tractability (relaxing the runtime).
From the approximability prospective, a logarithmic approximation factor can be found by using a simple greedy algorithm, and finding a sublogarithmic approximation factor is NP-hard~\cite{RaSa1997}.
The problem is in fact Log-APX-complete~\cite{EsPa2006}.
The Domination problem is notorious also in the theory of fixed-parameter tractability (see, \eg, \cite{DoFe1999,Niederm2006} for an introduction to parameterized complexity).
It was the first problem to be shown W[2]-complete~\cite{DoFe1999}, and it is hence unlikely to be FPT, \ie, unlikely to have an algorithm with runtime $f(k)n^c$ for $f$ a computable function, $k$ the size of an optimal solution, $c$ a constant, and $n$ the number of vertices of the input graph.
Similar results are known also for the connected domination problem~\cite{GuKh1998}.

The $p$-Center problem is known to be NP-hard on graphs.
However, for it, a simple and efficient factor~$2$ approximation algorithm exists~\cite{Gon1985}.
Furthermore, it is a best possible approximation algorithm in the sense that an approximation with factor less than~$2$ is proven to be NP-hard (see~\cite{Gon1985} for more details).
The NP-hardness of the Connected $p$-Center problem is shown in~\cite{YenChen2007}.

Recently, in~\cite{ChepoiEstell2007}, a new type of approximability result (call it a \emph{parameterized approximability} result) was obtained:
there exists a polynomial time algorithm which finds in an arbitrary graph~$G$ having a minimum $r$-dominating set~$D$ a set~$D'$ such that $|D'| \leq |D|$ and each vertex~$v \in V$ is within distance at most~$r(v) + 2 \delta$ from~$D'$, where $\delta$ is the hyperbolicity parameter of~$G$ (see~\cite{ChepoiEstell2007} for details).
We call such a~$D'$ an \emph{\( (r + 2 \delta) \)-dominating set} of~$G$.
Later, in~\cite{EdwaKennSani2016}, this idea was extended to the $p$-Center problem:
there is a quasi-linear time algorithm for the $p$-Center problem with an additive error less than or equal to six times the input graph's hyperbolicity (\ie, it finds a set~$C'$ with at most $p$~vertices such that $\max_{v \in V} d_G(v,C') \leq \min_{C \subseteq V, |C| \leq p} \max_{v \in V} d_G(v, C) + 6 \delta$).
We call such a~$C'$ a \emph{\( + 6\delta \)-approximation for the \( p \)-Center problem}.

In this paper, we continue the line of research started in \cite{ChepoiEstell2007} and~\cite{EdwaKennSani2016}.
Unfortunately, the results of~\cite{ChepoiEstell2007,EdwaKennSani2016} are hardly extendable to connected versions of the $r$-Domination and $p$-Center problems.
It remains an open question whether similar approximability results parameterized by the graph's hyperbolicity can be obtained for the Connected $r$-Domination and Connected $p$-Center problems.
Instead, we consider two other graph parameters:
the \emph{tree-breadth~\( \rho \)} and the \emph{cluster diameter~$\Delta$ in a layering partition} (formal definitions will be given in the next sections).
Both parameters (like the hyperbolicity) capture the metric tree-likeness of a graph (see, \eg, \cite{AbuAtaDragan2016} and papers cited therein).
As demonstrated in~\cite{AbuAtaDragan2016}, in many real-world networks, including Internet application networks, web networks, collaboration networks, social networks, biological networks, and others, as well as in many structured classes of graphs the parameters $\delta$, $\rho$, and~$\Delta$ are small constants.

We show here that, for a given $n$-vertex, $m$-edge graph~$G$, having a minimum $r$-dominating set~$D$ and a minimum connected $r$-dominating set~$C$:
\begin{itemize}
    \item
        an $(r + \Delta)$-dominating set~$D'$ with $|D'| \leq |D|$ can be computed in linear time;
    \item
        a connected $(r + 2 \Delta)$-dominating set~$C'$ with $|C'| \leq |C|$ can be computed in $\calO \big( m \, \alpha(n) \log \Delta \big)$ time (where $\alpha(n)$ is the inverse Ackermann function);
    \item
        a $+ \Delta$-approximation for the $p$-Center problem can be computed in linear time;
    \item
        a $+ 2 \Delta$-approximation for the connected $p$-Center problem can be computed in $\calO \big( m \, \alpha(n) \log \min(\Delta, p) \big)$ time.
\end{itemize}

Furthermore, given a tree-decomposition with breadth~$\rho$ for~$G$:
\begin{itemize}
    \item
        an $(r + \rho)$-dominating set~$D'$ with $|D'| \leq |D|$ can be computed in $\calO(nm)$ time;
    \item
        a connected $\big( r + 5 \rho \big)$-dominating set~$C'$ with $|C'| \leq |C|$ can be computed in $\calO(nm)$ time;
    \item
        a $+ \rho$-approximation for the $p$-Center problem can be computed in $\calO( nm \log n)$ time;
    \item
        a $+ 5 \rho$-approximation for the Connected $p$-Center problem can be computed in $\calO( nm \log n)$ time.
\end{itemize}

To compare these results with the results of~\cite{ChepoiEstell2007,EdwaKennSani2016}, notice that, for any graph~$G$, its hyperbolicity~$\delta$ is at most $\Delta$~\cite{AbuAtaDragan2016} and at most two times its tree-breadth~$\rho$~\cite{ChDr++2008}, and the inequalities are sharp.

Note that, for split graphs (graphs in which the vertices can be partitioned into a clique and an independent set), all three parameters are at most~$1$.
Additionally, as shown in~\cite{ChlebiChlebi2008}, there is (under reasonable assumptions) no polynomial-time algorithm to compute a sublogarithmic-factor approximation for the (Connected) Domination problem in split graphs.
Hence, there is no such algorithm even for constant $\delta$, $\rho$, and~$\Delta$.

One can extend this result to show that there is no polynomial-time algorithm~$\calA$ which computes, for any constant~$c$, a $+ c \log n$-approximation for split graphs.
Hence, there is no polynomial-time $+ c \Delta \log n$-approximation algorithm in general.
Consider a given split graph $G = (C \cup I, E)$ with $n$~vertices where $C$ induces a clique and $I$ induces an independent set.
Create a graph~$H = (C_H \cup I_H, E_H)$ by, first, making $n$~copies of~$G$.
Let $C_H = C_1 \cup C_2 \cup \ldots \cup C_n$ and $I_H = I_1 \cup I_2 \cup \ldots \cup I_n$.
Second, make the vertices in~$C_H$ pairwise adjacent.
Then, $C_H$ induces a clique and $I_H$ induces an independent set.
If there is such an algorithm~$\calA$, then $\calA$ produces a (connected) dominating set~$D_\calA$ for~$H$ which hast at most $2 c \log n$ more vertices that a minimum (connected) dominating set~$D$.
Thus, by pigeonhole principle, $H$ contains a clique~$C_i$ for which $|C_i \cap D_\calA| = |C_i \cap D|$.
Therefore, such an algorithm~$\calA$ would allow to solve the (Connected) Domination problem for split graphs in polynomial time.

\section{Preliminaries}
All graphs occurring in this paper are connected, finite, unweighted, undirected, without loops, and without multiple edges.
For a graph~$G = (V, E)$, we use $n = |V|$ and $m = |E|$ to denote the cardinality of the vertex set and the edge set of~$G$, respectively.

The \emph{length} of a path from a vertex~$v$ to a vertex~$u$ is the number of edges in the path.
The \emph{distance}~$d_G(u, v)$ in a graph~$G$ of two vertices $u$ and~$v$ is the length of a shortest path connecting $u$ and~$v$.
The distance between a vertex~$v$ and a set~$S \subseteq V$ is defined as $d_G(v, S) = \min_{u \in S} d_G(u, v)$.
For a vertex~$v$ of~$G$ and some positive integer~$r$, the set $N_G^r[v] = \big \{ \, u \mid d_G(u, v) \leq r \, \big \}$ is called the \emph{\( r \)-neighbourhood} of~$v$.
The \emph{eccentricity}~$\ecc_G(v)$ of a vertex~$v$ is $\max_{u \in V} d_G(u, v)$.
For a set~$S \subseteq V$, its eccentricity is $\ecc_G(S) = \max_{u \in V} d_G(u, S)$.

For some function $r \colon V \rightarrow \bbN$, a vertex~$u$ is \emph{\( r \)-dominated} by a vertex~$v$ (by a set~$S \subseteq V$), if $d_G(u, v) \leq r(u)$ ($d_G(u, S) \leq r(u)$, respectively).
A vertex set~$D$ is called an \emph{\( r \)-dominating set} of~$G$ if each vertex~$u \in V$ is $r$ dominated by~$D$.
Additionally, for some non-negative integer~$\phi$, we say a vertex is \emph{\( (r + \phi) \)-dominated} by a vertex~$v$ (by a set~$S \subseteq V$), if $d_G(u, v) \leq r(u) + \phi$ ($d_G(u, S) \leq r(u) + \phi$, respectively).
An \emph{\( (r + \phi) \)-dominating set} is defined accordingly.
For a given graph~$G$ and function~$r$, the \emph{(Connected) \( r \)-Domination} problem asks for the smallest (connected) vertex set~$D$ such that $D$ is an $r$-dominating set of~$G$.

The \emph{degree} of a vertex~$v$ is the number of vertices adjacent to it.
For a vertex set~$S$, let $G[S]$ denote the subgraph of~$G$ induced by~$S$.
A vertex set~$S$ is a \emph{separator} for two vertices $u$ and~$v$ in~$G$ if each path from $u$ to~$v$ contains a vertex~$s \in S$; in this case we say $S$ \emph{separates} $u$ from~$v$.

A \emph{tree-decomposition} of a graph~$G = (V, E)$ is a tree~$T$ with the vertex set~$\calB$ where each vertex of $T$, called bag, is a subset of~$V$ such that:
(i)~$V = \bigcup_{B \in \calB} B$, (ii)~for each edge~$uv \in E$, there is a bag~$B \in \calB$ with $u,v \in B$, and (iii)~for each vertex~$v \in V$, the bags containing $v$ induce a subtree of~$T$.
A tree-decomposition~$T$ of~$G$ has \emph{breadth~\( \rho \)} if, for each bag~$B$ of~$T$, there is a vertex~$v$ in~$G$ with $B \subseteq N_G^\rho[v]$.
The \emph{tree-breadth} of a graph~$G$ is $\rho$, written as $\tb(G) = \rho$, if $\rho$ is the minimal breadth of all tree-decomposition for~$G$.
A tree-decomposition~$T$ of~$G$ has \emph{length~\( \lambda \)} if, for each bag~$B$ of~$T$ and any two vertices~$u, v \in B$, $d_G(u, v) \leq \lambda$.
The \emph{tree-length} of a graph~$G$ is $\lambda$, written as $\tl(G) = \lambda$, if $\lambda$ is the minimal length of all tree-decomposition for~$G$.

For a rooted tree~$T$, let $\Lambda(T)$ denote the number of leaves of~$T$.
For the case when $T$ contains only one node, let $\Lambda(T) := 0$.
With $\alpha$, we denote the inverse Ackermann function (see, \eg,~\cite{CorLeiRivSte2009}).
It is well known that $\alpha$ grows extremely slowly.
For $x = 10^{80}$ (estimated number of atoms in the universe), $\alpha(x) \leq 4$.

\section{Using a Layering Partition}

The concept of a \emph{layering partition} was introduced in~\cite{BranChepDrag1999,ChepoiDragan2000}.
The idea is the following.
First, partition the vertices of a given graph~$G = (V, E)$ in distance layers~$L_i = \{ \, v \mid d_G(s, v) = i \, \}$ for a given vertex~$s$.
Second, partition each layer~$L_i$ into \emph{clusters} in such a way that two vertices $u$ and~$v$ are in the same cluster if and only if they are connected by a path only using vertices in the same or upper layers.
That is, $u$ and~$v$ are in the same cluster if and only if, for some~$i$, $\{ u, v \} \subseteq L_i$ and there is a path~$P$ from~$u$ to~$v$ in~$G$ such that, for all $j < i$, $P \cap L_j = \emptyset$.
Note that each cluster~$C$ is a set of vertices of~$G$, \ie, $C \subseteq V$, and all clusters are pairwise disjoint.
The created clusters form a rooted tree~$\calT$ with the cluster~$\{ s \}$ as the root where each cluster is a node of~$\calT$ and two clusters $C$ and~$C'$ are adjacent in~$\calT$ if and only if $G$ contains an edge~$uv$ with $u \in C$ and~$v \in C'$.
Figure~\ref{fig:LayPartEx} gives an example for such a partition.
A layering partition of a graph can be computed in linear time~\cite{ChepoiDragan2000}.

\begin{figure}
    [htb]
    \centering

    \subfloat[]
    [%
        A graph $G$.
    ]
    {%
        \centering
        \includegraphics[]{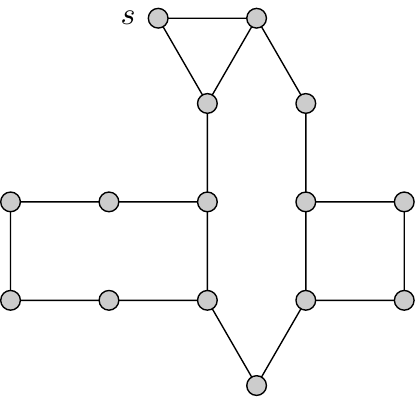}
        \label{fig:LayPartExG}
    }
    \hfil
    \subfloat[]
    [%
        A layering partition~$\calT$ for~$G$.
    ]
    {%
        \centering
        \includegraphics[]{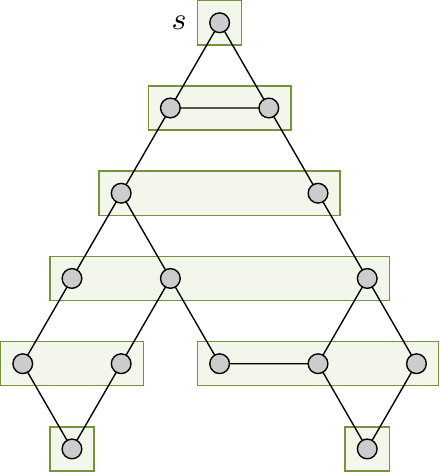}
        \label{fig:LayPartExT}
    }

    \caption
    {%
        Example of a layering partition.
        A given graph~$G$~\protect\subref{fig:LayPartExG} and the layering partition of~$G$ generated when starting at vertex~$s$~\protect\subref{fig:LayPartExT}.
        Example taken from~\cite{ChepoiDragan2000}.
    }
    \label{fig:LayPartEx}
\end{figure}

For the remainder of this section, assume that we are given a graph~$G = (V, E)$ and a layering partition~$\calT$ of~$G$ for an arbitrary start vertex.
We denote the largest diameter of all clusters of~$\calT$ as~$\Delta$, \ie, $\Delta := \max \big \{ \, d_G(x, y) \mid \text{$x, y$ are in a cluster~$C$ of~$\calT$} \, \big \}$.
For two vertices $u$ and~$v$ of~$G$ contained in the clusters $C_u$ and~$C_v$ of~$\calT$, respectively, we define $d_\calT(u, v) := d_\calT(C_u, C_v)$.

\begin{lemma}
    \label{lem:LayPartVertDist}
For all vertices \( u \) and~\( v \) of~\( G \), \( d_\calT(u, v) \leq d_G(u, v) \leq d_\calT(u, v) + \Delta \).
\end{lemma}

\begin{proof}
Clearly, by construction of a layering partition, $d_{\calT}(u, v) \leq d_G(u, v)$ for all vertices $u$ and~$v$ of~$G$.

Next, let $C_u$ and~$C_v$ be the clusters containing $u$ and~$v$, respectively.
Note that $\calT$ is a rooted tree.
Let $C'$ be the lowest common ancestor of $C_u$ and~$C_v$.
Therefore, $d_\calT(u, v) = d_\calT(u, C') + d_\calT(C', v)$.
By construction of a layering partition, $C'$ contains a vertex~$u'$ and vertex~$v'$ such that $d_G(u, u') = d_\calT(u, u')$ and $d_G(v, v') = d_\calT(v, v')$.
Since the diameter of each cluster is at most~$\Delta$, $d_G(u, v) \leq d_\calT(u, u') + \Delta + d_\calT(v, v') = d_{\calT}(u, v) + \Delta$.
\qed
\end{proof}

Theorem~\ref{theo:rDeltaDom} below shows that we can use the layering partition~$\calT$ to compute an $(r + \Delta)$-dominating set for~$G$ in linear time which is not larger than a minimum $r$-dominating set for~$G$.
This is done by finding a minimum $r$-dominating set of~$\calT$ where, for each cluster~$C$ of~$\calT$, $r(C)$ is defined as~$\min_{v \in C} r(v)$.

\begin{theorem}
    \label{theo:rDeltaDom}
Let \( D \) be a minimum \( r \)-dominating set for a given graph~\( G \).
An \( (r + \Delta) \)-dominating set~\( D' \) for~\( G \) with \( |D'| \leq |D| \) can be computed in linear time.
\end{theorem}

\begin{proof}
First, create a layering partition~$\calT$ of~$G$ and, for each cluster~$C$ of~$\calT$, set $r(C) := \min_{v \in C} r(v)$.
Second, find a minimum $r$-dominating set~$\calS$ for~$\calT$, \ie, a set~$\calS$ of clusters such that, for each cluster~$C$ of~$\calT$, $d_\calT(C, \calS) \leq r(C)$.
Third, create a set~$D'$ by picking an arbitrary vertex of~$G$ from each cluster in~$\calS$.
All three steps can be performed in linear time, including the computation of~$\calS$ (see~\cite{BranChepDrag1998}).

Next, we show that $D'$ is an $(r + \Delta)$-dominating set for~$G$.
By construction of~$\calS$, each cluster~$C$ of~$\calT$ has distance at most~$r(C)$ to~$\calS$ in~$\calT$.
Thus, for each vertex~$u$ of~$G$, $\calS$ contains a cluster~$C_\calS$ with $d_\calT(u, C_\calS) \leq r(u)$.
Additionally, by Lemma~\ref{lem:LayPartVertDist}, $d_G(u, v) \leq r(u) + \Delta$ for any vertex~$v \in C_\calS$.
Therefore, for any vertex~$u$, $d_G(u, D') \leq r(u) + \Delta$, \ie, $D'$ is an $(r + \Delta)$-dominating set for~$G$.

It remains to show that $|D'| \leq |D|$.
Let $\calD$ be the set of clusters of~$\calT$ that contain a vertex of~$D$.
Because $D$ is an $r$-dominating set for~$G$, it follows from Lemma~\ref{lem:LayPartVertDist} that $\calD$ is an $r$-dominating set for~$\calT$.
Clearly, since clusters are pairwise disjoint, $|\calD| \leq |D|$.
By minimality of~$\calS$, $|\calS| \leq |\calD|$ and, by construction of~$D'$, $|D'| = |\calS|$.
Therefore, $|D'| \leq |D|$.
\qed
\end{proof}

We now show how to construct a connected $(r + 2 \Delta)$-dominating set for~$G$ using~$\calT$ in such a way that the set created is not larger than a minimum connected $r$-dominating set for~$G$.
For the remainder of this section, let $D_r$ be a minimum connected $r$-dominating set of~$G$ and let, for each cluster~$C$ of~$\calT$, $r(C)$ be defined as above.
Additionally, we say that a subtree~$T'$ of some tree~$T$ is an \emph{\( r \)-dominating subtree of~\( T \)} if the nodes (clusters in case of a layering partition) of~$T'$ form a connected $r$-dominating set for~$T$.

The first step of our approach is to construct a minimum $r$-dominating subtree~$T_r$ of~$\calT$.
Such a subtree~$T_r$ can be computed in linear time~\cite{Dragan1993}.
Lemma~\ref{lem:TrCardinality} below shows that $T_r$ gives a lower bound for the cardinality of~$D_r$.

\begin{lemma}
    \label{lem:TrCardinality}
If \( T_r \) contains more than one cluster, each connected \( r \)-dominating set of~\( G \) intersects all clusters of~\( T_r \).
Therefore, \( |T_r| \leq |D_r| \).
\end{lemma}

\begin{proof}
Let $D$ be an arbitrary connected $r$-dominating set of~$G$.
Assume that $T_r$ has a cluster~$C$ such that $C \cap D = \emptyset$.
Because $D$ is connected, the subtree of~$\calT$ induced by the clusters intersecting~$D$ is connected, too.
Thus, if $D$ intersects all leafs of~$T_r$, then it intersects all clusters of~$T_r$.
Hence, we can assume, without loss of generality, that $C$ is a leaf of~$T_r$.
Because $T_r$ has at least two clusters and by minimality of~$T_r$, $\calT$~contains a cluster~$C'$ such that $d_\calT(C', C) = d_\calT(C', T_r) = r(C')$.
Note that each path in~$G$ from a vertex in~$C'$ to a vertex in~$D$ intersects~$C$.
Therefore, by Lemma~\ref{lem:LayPartVertDist}, there is a vertex~$u \in C'$ with $r(u) = d_\calT(u, C) < d_\calT(u, D) \leq d_G(u, D)$.
That contradicts with $D$ being an $r$-dominating set.

Because any $r$-dominating set of~$G$ intersects each cluster of~$T_r$ and because these clusters are pairwise disjoint, it follows that $|T_r| \leq |D_r|$.
\qed
\end{proof}

As we show later in Corollary~\ref{cor:TdeltaDomSet}, each connected vertex set~$S \subseteq V$ that intersects each cluster of~$T_r$ gives an $(r + \Delta)$-dominating set for~$G$.
It follows from Lemma~\ref{lem:TrCardinality} that, if such a set~$S$ has minimum cardinality, $|S| \leq |D_r|$.
However, finding a minimum cardinality connected set intersecting each cluster of a layering partition (or of a subtree of it) is as hard as finding a minimum Steiner tree.

The main idea of our approach is to construct a minimum $(r + \delta)$-dominating subtree~$T_\delta$ of~$\calT$ for some integer~$\delta$.
We then compute a small enough connected set~$S_\delta$ that intersects all cluster of~$T_\delta$.
By trying different values of~$\delta$, we eventually construct a connected set~$S_\delta$ such that $|S_\delta| \leq |T_r|$ and, thus, $|S_\delta| \leq |D_r|$.
Additionally, we show that $S_\delta$ is a connected $(r + 2 \Delta)$-dominating set of~$G$.

For some non-negative integer~$\delta$, let $T_\delta$ be a minimum $(r + \delta)$-dominating subtree of~$\calT$.
Clearly, $T_0 = T_r$.
The following two lemmas set an upper bound for the maximum distance of a vertex of~$G$ to a vertex in a cluster of~$T_\delta$ and for the size of~$T_\delta$ compared to the size of~$T_r$.

\begin{lemma}
    \label{lem:TdeltaDist}
For each vertex~\( v \) of~\( G \), \( d_\calT(v, T_\delta) \leq r(v) + \delta \).
\end{lemma}

\begin{proof}
Let $C_v$ be the cluster of~$\calT$ containing~$v$ and let $C$ be the cluster of~$T_\delta$ closest to~$C_v$ in~$\calT$.
By construction of~$T_\delta$, $d_\calT(v, C) = d_\calT(C_v, C) \leq r(C_v) + \delta \leq r(v) + \delta$.
\qed
\end{proof}

Because the diameter of each cluster is at most~$\Delta$, Lemma~\ref{lem:LayPartVertDist} and Lemma~\ref{lem:TdeltaDist} imply the following.

\begin{corollary}
    \label{cor:TdeltaDomSet}
If a vertex set intersects all clusters of~\( T_\delta \), it is an \( \big( r + (\delta + \Delta) \big) \)-dominating set of~\( G \).
\end{corollary}

\begin{lemma}
    \label{lem:TdeltaTrCardinality}
\( |T_{\delta}| \leq |T_r| - \delta \cdot \Lambda(T_\delta) \).
\end{lemma}

\begin{proof}
First, consider the case when $T_\delta$ contains only one cluster, \ie, $|T_\delta| = 1$.
Then, $\Lambda(T_\delta) = 1$ and, thus, the statement clearly holds.
Next, let $T_\delta$ contain more than one cluster, let $C_u$ be an arbitrary leaf of~$T_\delta$, and let $C_v$ be a cluster of~$T_r$ with maximum distance to~$C_u$ such that $C_u$ is the only cluster on the shortest path from $C_u$ to~$C_v$ in~$T_r$, \ie, $C_v$ is not in~$T_\delta$.
Due to the minimality of~$T_\delta$, $d_{T_r}(C_u, C_v) = \delta$.
Thus, the shortest path from $C_u$ to~$C_v$ in~$T_r$ contains $\delta$~clusters (including~$C_v$) which are not in~$T_\delta$.
Therefore, $|T_{\delta}| \leq |T_r| - \delta \cdot \Lambda(T_\delta)$.
\qed
\end{proof}

Now that we have constructed and analysed~$T_\delta$, we show how to construct~$S_\delta$.
First, we construct a set of shortest paths such that each cluster of~$T_\delta$ is intersected by exactly one path.
Second, we connect these paths with each other to from a connected set using an approach which is similar to Kruskal's algorithm for minimum spanning trees.

Let $\calL = \big \{ C_1, C_2, \ldots, C_\lambda \big \}$ be the leaf clusters of~$T_\delta$ (excluding the root) with either $\lambda = \Lambda(T_\delta) - 1$ if the root of~$T_\delta$ is a leaf, or with $\lambda = \Lambda(T_\delta)$ otherwise.
We construct a set~$\calP = \big \{ P_1, P_2, \ldots, P_\lambda \big \}$ of paths as follows.
Initially, $\calP$ is empty.
For each cluster~$C_i \in \calL$, in turn, find the ancestor~$C_i'$ of~$C_i$ which is closest to the root of~$T_\delta$ and does not intersect any path in~$\calP$ yet.
If we assume that the indices of the clusters in~$\calL$ represent the order in which they are processed, then $C_1'$ is the root of~$T_\delta$.
Then, select an arbitrary vertex~$v$ in~$\calC_i$ and find a shortest path~$P_i$ in~$G$ form~$v$ to~$C_i'$.
Add $P_i$ to~$\calP$ and continue with the next cluster in~$\calL$.
Figure~\ref{fig:LayPartPaths} gives an example.

\begin{figure}
    [htb]
    \centering
    \includegraphics{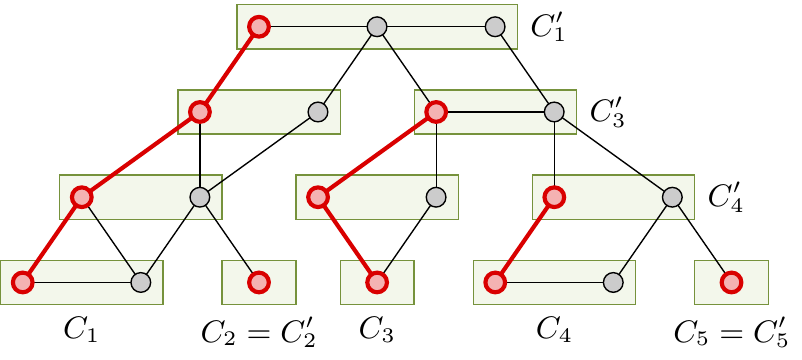}
    \caption
    [%
        Example for the set~$\calP$ for a subtree of a layering partition.
    ]
    {%
        Example for the set~$\calP$ for a subtree of a layering partition.
        Paths are shown in red.
        Each path~$P_i$, with $1 \leq i \leq 5$, starts in the leaf~$C_i$ and ends in the cluster~$C_i'$.
        For $i = 2$ and $i = 5$, $P_i$ contains only one vertex.
    }
    \label{fig:LayPartPaths}
\end{figure}

\begin{lemma}
    \label{lem:clusterPath}
For each cluster~\( C \) of~\( T_\delta \), there is exactly one path~\( P_i \in \calP \) intersecting~\( C \).
Additionally, \( C \) and~\( P_i \) share exactly one vertex, \ie, \( |C \cap P_i| = 1 \).
\end{lemma}

\begin{proof}
Observe that, by construction of a layering partition, each vertex in a cluster~$C$ is adjacent to some vertex in the parent cluster of~$C$.
Therefore, a shortest path~$P$ in~$G$ from~$C$ to any of its ancestors~$C'$ only intersects clusters on the path from $C$ to~$C'$ in~$\calT$ and each cluster shares only one vertex with~$P$.
It remains to show that each cluster intersects exactly one path.

Without loss of generality, assume that the indices of clusters in~$\calL$ and paths in~$\calP$ represent the order in which they are processed and created, \ie, assume that the algorithms first creates $P_1$ which starts in~$C_1$, then $P_2$ which starts in~$C_2$, and so on.
Additionally, let $\calL_i = \{ C_1, C_2, \ldots, C_i \}$ and $\calP_i = \{ P_1, P_2, \ldots, P_i \}$.

To proof that each cluster intersects exactly one path, we show by induction over~$i$ that, if a cluster~$C_i$ of~$T_\delta$ satisfies the statement, then all ancestors of~$C_i$ satisfy it, too.
Thus, if $C_\lambda$ satisfies the statement, each cluster satisfies it.

First, consider $i = 1$.
Clearly, since $P_1$ is the first path, $P_1$ connects the leaf~$C_1$ with the root of~$T_\delta$ and no cluster intersects more than one path at this point.
Therefore, the statement is true for~$C_1$ and each of its ancestors.

Next, assume that $i > 1$ and that the statement is true for each cluster in~$\calL_{i - 1}$ and their respective ancestors.
Then, the algorithm creates $P_i$ which connects the leaf~$C_i$ with the cluster~$C'_i$.
Assume that there is a cluster~$C$ on the path from $C_i$ to~$C_i'$ in~$\calT$ such that $C$ intersects a path~$P_j$ with $j < i$.
Clearly, $C_i'$ is an ancestor of~$C$.
Thus, by induction hypothesis, $C_i'$ is also intersected by some path~$P \neq P_i$.
This contradicts with the way $C_i'$ is selected by the algorithm.
Therefore, each cluster on the path from $C_i$ to~$C_i'$ in~$\calT$ only intersects $P_i$ and $P_i$ does not intersect any other clusters.

Because $i > 1$, $C'_i$ has a parent cluster~$C''$ in~$T_\delta$ that is intersected by a path~$P_j$ with $j < i$.
By induction hypothesis, each ancestor of~$C''$ is intersected by a path in~$\calP_{i - 1}$.
Therefore, each ancestor of~$C_i$ is intersected by exactly one path in~$\calP_i$.
\qed
\end{proof}

Next, we use the paths in~$\calP$ to create the set~$S_\delta$.
As first step, let $S_\delta := \bigcup_{P_i \in \calP} P_i$.
Later, we add more vertices into~$S_\delta$ to ensure it is a connected set.

Now, create a partition~$\calV = \big \{ V_1, V_2, \ldots, V_\lambda \big \}$ of~$V$ such that, for each~$i$, $P_i \subseteq V_i$, $V_i$ is connected, and $d_G(v, P_i) = \min_{P \in \calP} d_G(v, P)$ for each vertex~$v \in V_i$.
That is, $V_i$ contains the vertices of~$G$ which are not more distant to~$P_i$ in~$G$ than to any other path in~$\calP$.
Additionally, for each vertex~$v \in V$, set $P(v) := P_i$ if and only if $v \in V_i$ (\ie, $P(v)$ is the path in~$\calP$ which is closest to~$v$) and set~$d(v) := d_G \big( v, P(v) \big)$.
Such a partition as well as $P(v)$ and~$d(v)$ can be computed by performing a BFS on~$G$ starting at all paths~$P_i \in \calP$ simultaneously.
Later, the BFS also allows us to easily determine the shortest path from $v$ to~$P(v)$ for each vertex~$v$.

To manage the subsets of~$\calV$, we use a Union-Find data structure such that, for two vertices $u$ and~$v$, $\mathrm{Find}(u) = \mathrm{Find}(v)$ if and only if $u$ and~$v$ are in the same set of~$\calV$.
A Union-Find data structure additionally allows us to easily join two set of~$\calV$ into one by performing a single $\mathrm{Union}$ operation.
Note that, whenever we join two sets of~$\calV$ into one, $P(v)$ and~$d(v)$ remain unchanged for each vertex~$v$.

Next, create an edge set~$E' = \{ \, uv \mid \mathrm{Find}(u) \neq \mathrm{Find}(v) \, \}$, \ie, the set of edges~$uv$ such that $u$ and~$v$ are in different sets of~$\calV$.
Sort $E'$ in such a way that an edge~$uv$ precedes an edge~$xy$ only if $d(u) + d(v) \leq d(x) + d(y)$.

The last step to create $S_\delta$ is similar to Kruskal's minimum spanning tree algorithm.
Iterate over the edges in~$E'$ in increasing order.
If, for an edge~$uv$, $\mathrm{Find}(u) \neq \mathrm{Find}(v)$, \ie, if $u$ and~$v$ are in different sets of~$\calV$, then join these sets into one by performing $\mathrm{Union}(u, v)$, add the vertices on the shortest path from $u$ to~$P(u)$ to~$S_\delta$, and add the vertices on the shortest path from $v$ to~$P(v)$ to~$S_\delta$.
Repeat this, until $\calV$ contains only one set, \ie, until $\calV = \{ V \}$.

Algorithm~\ref{algo:Tdelta} below summarises the steps to create a set~$S_\delta$ for a given subtree of a layering partition subtree~$T_\delta$.

\begin{algorithm}
    [!htb]
    \caption
    {%
        Computes a connected vertex set that intersects each cluster of a given layering partition.
    }
    \label{algo:Tdelta}

\KwIn
{%
    A graph~$G = (V, E)$ and a subtree~$T_\delta$ of some layering partition of~$G$.
}

\KwOut
{%
    A connected set~$S_\delta \subseteq V$ that intersects each cluster of~$T_\delta$ and contains at most $|T_\delta| + \big( \Lambda(T_\delta) - 1 \big) \cdot \Delta$ vertices.
}

Let $\calL = \big \{ C_1, C_2, \ldots, C_\lambda \big \}$ be the set of clusters excluding the root that are leaves of~$T_\delta$.

Create an empty set~$\calP$.
\label{line:createEmptyP}

\ForEach
{%
    cluster~\( C_i \in \calL \)
}
{%
    Select an arbitrary vertex~$v \in C_i$.

    Find the highest ancestor~$C_i'$ of~$C_i$ (\ie, the ancestor which is closest to the root of~$T_\delta$) that is not flagged.

    Find a shortest path~$P_i$ from~$v$ to an ancestor of~$v$ in~$C_i'$ (\ie, a shortest path from $C_i$ to~$C_i'$ in~$G$ that contains exactly one vertex of each cluster of the corresponding path in~$T_\delta$).

    Add $P_i$ to~$\calP$.

    Flag each cluster intersected by~$P_i$.
    \label{line:flagPiClusters}
}

Create a set $S_\delta := \bigcup_{P_i \in \calP} P_i$.
\label{line:addPtoS}

Perform a BFS on~$G$ starting at all paths~$P_i \in \calP$ simultaneously.
This results in a partition $\calV = \big \{ V_1, V_2, \ldots, V_\lambda \big \}$ of~$V$ with $P_i \subseteq V_i$ for each $P_i \in \calP$.
For each vertex~$v$, set $P(v) := P_i$ if and only if $v \in V_i$ and let $d(v) := d_G(v, P(v))$.
\label{line:performBFS}

Create a Union-Find data structure and add all vertices of~$G$ such that $\mathrm{Find}(v) = i$ if and only if $v \in V_i$.
\label{line:initUnionFind}

Determine the edge set~$E' = \{ \, uv \mid \mathrm{Find}(u) \neq \mathrm{Find}(v) \, \}$.
\label{line:determineEprime}

Sort $E'$ such that $uv \leq xy$ if and only if $d(u) + d(v) \leq d(x) + d(y)$.
Let $\langle e_1, e_2, \ldots, e_{|E'|} \rangle$ be the resulting sequence.
\label{line:sortEprime}

\For
{%
    \( i := 1 \) \KwTo \( |E'| \)%
    \label{line:EprimeLoop}
}
{%
    Let $uv = e_i$.

    \If
    {%
        \( \mathrm{Find}(u) \neq \mathrm{Find}(v) \)
    }
    {%
        Add the shortest path from $u$ to~$P(u)$ to~$S_\delta$.
        \label{line:addPuToS}

        Add the shortest path from $v$ to~$P(v)$ to~$S_\delta$.
        \label{line:addPvToS}

        $\mathrm{Union}(u, v)$
        \label{line:unionSets}
    }
}

Output~$S_\delta$.
\end{algorithm}

\begin{lemma}
    \label{lem:SdeltaCardinality}
For a given graph~\( G \) and a given subtree~\( T_\delta \) of some layering partition of~\( G \), Algorithm~\ref{algo:Tdelta} constructs, in \( \calO \big( m \, \alpha(n) \big) \)~time, a connected set~\( S_\delta \) with \( |S_\delta| \leq |T_\delta| + \Delta \cdot \Lambda(T_\delta) \) which intersects each cluster of~\( T_\delta \).
\end{lemma}

\begin{proof}
    [Correctness]
First, we show that $S_\delta$ is connected at the end of the algorithm.
To do so, we show by induction that, at any time, $S_\delta \cap V'$ is a connected set for each set~$V' \in \calV$.
Clearly, when $\calV$ is created, for each set~$V_i \in \calV$, $S_\delta \cap V_i = P_i$.
Now, assume that the algorithm joins the set $V_u$ and~$V_v$ in~$\calV$ into one set based on the edge~$uv$ with $u \in V_u$ and $v \in V_v$.
Let $S_u = S_\delta \cap V_u$ and $S_v = S_\delta \cap V_v$.
Note that $P(u) \subseteq S_u$ and $P(v) \subseteq S_v$.
The algorithm now adds all vertices to~$S_\delta$ which are on a path from $P(u)$ to~$P(v)$.
Therefore, $S_\delta \cap (V_u \cup V_v)$ is a connected set.
Because $\calV = \{ V \}$ at the end of the algorithm, $S_\delta$ is connected eventually.
Additionally, since $P_i \subseteq S_\delta$ for each $P_i \in \calP$, it follows that $S_\delta$ intersects each cluster of~$T_\delta$.

Next, we show that the cardinality of $S_\delta$ is at most $|T_\delta| + \Delta \cdot \Lambda(T_\delta)$.
When first created, the set~$S_\delta$ contains all vertices of all paths in~$\calP$.
Therefore, by Lemma~\ref{lem:clusterPath}, $|S_\delta| = \sum_{P_i \in \calP} |P_i| = |T_\delta|$.
Then, each time two sets of~$\calV$ are joined into one set based on an edge~$uv$, $S_\delta$ is extended by the vertices on the shortest paths from $u$ to~$P(u)$ and from $v$ to~$P(v)$.
Therefore, the size of~$S_\delta$ increases by $d(u) + d(v)$, \ie, $|S_\delta| := |S_\delta| + d(u) + d(v)$.
Let $X$ denote the set of all edges used to join two sets of~$\calV$ into one at some point during the algorithm.
Note that $|X| = |\calP| - 1 \leq \Lambda(T_\delta)$.
Therefore, at the end of the algorithm,
\[
    |S_\delta|
        = \sum_{\mathclap{P_i \in \calP}} |P_i| + \sum_{\mathclap{uv \in X}} \big( d(u) + d(v) \big)
        \leq |T_\delta| + \Lambda(T_\delta) \cdot \max_{\mathclap{uv \in X}} \big( d(u) + d(v) \big).
\]

\begin{claim}
For each edge~\( uv \in X \), \( d(u) + d(v) \leq \Delta \).
\end{claim}

\begin{proof}[Claim]
To represent the relations between paths in~$\calP$ and vertex sets in~$\calV$, we define a function~$f \colon \calP \rightarrow \calV$ such that $f(P_i) = V_j$ if and only if $P_i \subseteq V_j$.
Directly after constructing~$\calV$, $f$ is a bijection with $f(P_i) = V_i$.
At the end of the algorithm, after all sets of~$\calV$ are joined into one, $f(P_i) = V$ for all $P_i \in \calP$.

Recall the construction of~$\calP$ and assume that the indices of the paths in~$\calP$ represent the order in which they are created.
Assume that $i > 1$.
By construction, the path~$P_i \in \calP$ connects the leaf~$C_i$ with the cluster~$C'_i$ in~$T_\delta$.
Because $i > 1$, $C'_i$ has a parent cluster in~$T_\delta$ that is intersected by a path~$P_j \in \calP$ with $j < i$.
We define~$P_j$ as the \emph{parent} of~$P_i$.
By Lemma~\ref{lem:clusterPath}, this parent~$P_j$ is unique for each $P_i \in \calP$ with~$i > 1$.
Based on this relation between paths in~$\calP$, we can construct a rooted tree~$\bbT$ with the node set~$\{ \, x_i \mid P_i \in \calP \, \}$ such that each node~$x_i$ represents the path~$P_i$ and $x_j$ is the parent of~$x_i$ if and only if $P_j$ is the parent of~$P_i$.

Because each node of~$\bbT$ represents a path in~$\calP$, $f$~defines a colouring for the nodes of~$\bbT$ such that $x_i$ and~$x_j$ have different colours if and only if $f(P_i) \neq f(P_j)$.
As long as $|\calV| > 1$, $\bbT$ contains two adjacent nodes with different colours.
Let $x_i$ and~$x_j$ be these nodes with $j < i$ and let $P_i$ and~$P_j$ be the corresponding paths in~$\calP$.
Note that $x_j$ is the parent of~$x_i$ in~$\bbT$ and, hence, $P_j$ is the parent of~$P_i$.
Therefore, $P_i$ ends in a cluster~$C'_i$ which has a parent cluster~$C$ that intersects~$P_j$.
By properties of layering partitions, it follows that $d_G(P_i, P_j) \leq \Delta + 1$.
Recall that, by construction, $d(v) = \min_{P \in \calP} d_G(v, P)$ for each vertex~$v$.
Thus, for each edge~$uv$ on a shortest path from $P_i$ to~$P_j$ in~$G$ (with $u$ being closer to~$P_i$ than to~$P_j$), $d(u) + d(v) \leq d_G(u, P_i) + d_G(v, P_j) \leq \Delta$.
Therefore, because $f(P_i) \neq f(P_j)$, there is an edge~$uv$ on a shortest path from $P_i$ to~$P_j$ such that $f \big( P(u) \big) \neq f \big( P(v) \big)$ and $d(u) + d(v) \leq \Delta$.
\qedClaim
\end{proof}

From the claim above, it follows that, as long as $\calV$ contains multiple sets, there is an edge~$uv \in E'$ such that $d(u) + d(v) \leq \Delta$ and $\mathrm{Find}(u) \neq \mathrm{Find}(v)$.
Therefore, $\max_{uv \in X} \big( d(u) + d(v) \big) \leq \Delta$ and $|S_\delta| \leq |T_\delta| + \big( \Lambda(T_\delta) - 1 \big) \cdot \Delta$.
\qed
\end{proof}

\begin{proof}
    [Complexity]
First, the algorithm computes~$\calP$ (line~\ref{line:createEmptyP} to line~\ref{line:flagPiClusters}).
If the parent of each vertex from the original BFS that was used to construct~$\calT$ is still known, $\calP$ can be constructed in $\calO(n)$ total time.
After picking a vertex~$v$ in $C_i$, simply follow the parent pointers until a vertex in~$C_i'$ is reached.
Computing $\calV$ as well as $P(v)$ and $d(v)$ for each vertex~$v$ of~$G$ (line~\ref{line:performBFS}) can be done with single BFS and, thus, requires at most $\calO(n + m)$ time.

Recall that, for a Union-Find data structure storing $n$~elements, each operation requires at most $\calO \big( \alpha(n) \big)$ amortised time.
Therefore, initialising such a data structure to store all vertices (line~\ref{line:initUnionFind}) and computing $E'$ (line~\ref{line:determineEprime}) requires at most $\calO \big( m \, \alpha(n) \big)$ time.
Note that, for each vertex~$v$, $d(v) \leq |V|$.
Thus, sorting~$E'$ (line~\ref{line:sortEprime}) can be done in linear time using counting sort.
When iterating over $E'$ (line~\ref{line:EprimeLoop} to line~\ref{line:unionSets}), for each edge~$uv \in E'$, the $\mathrm{Find}$-operation is called twice and the $\mathrm{Union}$-operation is called at most once.
Thus, the total runtime for all these operations is at most $\calO \big( m \, \alpha(n) \big)$.

Let $P_u = \{ u, \ldots, x, y, \ldots, p \}$ be the shortest path in~$G$ from a vertex~$u$ to~$P(u)$.
Assume that $y$ has been added to~$S_\delta$ in a previous iteration.
Thus, $\{ y, \ldots, p \} \subseteq S_\delta$ and, when adding $P_u$ to~$S_\delta$, the algorithm only needs to add $\{ u, \ldots, x \}$.
Therefore, by using a simple binary flag to determine if a vertex is contained in~$S_\delta$, constructing~$S_\delta$ (line~\ref{line:addPtoS}, line~\ref{line:addPuToS}, and line~\ref{line:addPvToS}) requires at most $\calO(n)$ time.

In total, Algorithm~\ref{algo:Tdelta} runs in $\calO \big( m \, \alpha(n) \big)$ time.
\qed
\end{proof}

Because, for each integer~$\delta \geq 0$, $|S_\delta| \leq |T_\delta| + \Delta \cdot \Lambda(T_\delta)$ (Lemma~\ref{lem:SdeltaCardinality}) and $|T_{\delta}| \leq |T_r| - \delta \cdot \Lambda(T_\delta)$ (Lemma~\ref{lem:TdeltaTrCardinality}), we have the following.

\begin{corollary}
    \label{cor:SdeltaCardinality}
For each \( \delta \geq \Delta \), \( |S_\delta| \leq |T_r| \) and, thus, \( |S_\delta| \leq |D_r| \).
\end{corollary}

To the best of our knowledge, there is no algorithm known that computes $\Delta$ in less than $\calO(nm)$ time.
Additionally, under reasonable assumptions, computing the diameter or radius of a general graph requires~$\Omega \big( n^2 \big)$ time~\cite{AbboWillWang2016}.
We conjecture that the runtime for computing $\Delta$ for a given graph has a similar lower bound.

To avoid the runtime required for computing~$\Delta$, we use the following approach shown in Algorithm~\ref{algo:con2DeltaDom} below.
First, compute a layering partition~$\calT$ and the subtree~$T_r$.
Second, for a certain value of~$\delta$, compute~$T_\delta$ and perform Algorithm~\ref{algo:Tdelta} on it.
If the resulting set~$S_\delta$ is larger than~$T_r$ (\ie, $|S_\delta| > |T_r|$), increase~$\delta$; otherwise, if $|S_\delta| \leq |T_r|$, decrease~$\delta$.
Repeat the second step with the new value of~$\delta$.

One strategy to select values for~$\delta$ is a classical binary search over the number of vertices of~$G$.
In this case, Algorithm~\ref{algo:Tdelta} is called up-to $\calO(\log n)$ times.
Empirical analysis~\cite{AbuAtaDragan2016}, however, have shown that $\Delta$ is usually very small.
Therefore, we use a so-called \emph{one-sided} binary search.

Consider a sorted sequence~$\langle x_1, x_2, \ldots, x_{n} \rangle$ in which we search for a value~$x_p$.
We say the value~$x_i$ is at position~$i$.
For a one-sided binary search, instead of starting in the middle at position~$n/2$, we start at position~$1$.
We then processes position~$2$, then position~$4$, then position~$8$, and so on until we reach position~$j = 2^i$ and, next, position~$k = 2^{i+1}$ with $x_j < x_p \leq x_k$.
Then, we perform a classical binary search on the sequence~$\langle x_{j + 1}, \ldots, x_k \rangle$.
Note that, because $x_j < x_p \leq x_k$, $2^i < p \leq 2^{i+1}$ and, hence, $j < p \leq k < 2p$.
Therefore, a one-sided binary search requires at most $\calO(\log p)$ iterations to find~$x_p$.

Because of Corollary~\ref{cor:SdeltaCardinality}, using a one-sided binary search allows us to find a value~$\delta \leq \Delta$ for which $|S_\delta| \leq |T_r|$ by calling Algorithm~\ref{algo:Tdelta} at most $\calO(\log \Delta)$ times.
Algorithm~\ref{algo:con2DeltaDom} below implements this approach.

\SetKwFor{OSBS}{One-Sided Binary Search}{}{}

\begin{algorithm}
    [htb]
    \caption
    {%
        Computes a connected $(r + 2 \Delta)$-dominating set for a given graph~$G$.
    }
    \label{algo:con2DeltaDom}

\KwIn
{%
    A graph~$G = (V, E)$ and a function~$r \colon V \rightarrow \bbN$.
}

\KwOut
{%
    A connected $(r + 2 \Delta)$-dominating set~$D$ for~$G$ with $|D| \leq |D_r|$.
}

Create a layering partition~$\calT$ of~$G$.

For each cluster~$C$ of~$\calT$, set $r(C) := \min_{v \in C} r(v)$.

Compute a minimum $r$-dominating subtree~$T_r$ for~$\calT$ (see~\cite{Dragan1993}).
\label{line:compTr}

\OSBS
{%
    over~\( \delta \), starting with \( \delta = 0 \)
}
{%
    Create a minimum $\delta$-dominating subtree~$T_\delta$ of~$T_r$ (\ie, $T_\delta$ is a minimum $(r + \delta)$-dominating subtree for~$\calT$).
    \label{line:compDeltaDom}

    Run Algorithm~\ref{algo:Tdelta} on~$T_\delta$ and let the set~$S_\delta$ be the corresponding output.

    \If
    {%
        \( |S_\delta| \leq |T_r| \)%
        \label{line:checkSdelta}%
    }
    {%
        Decrease~$\delta$.
    }
    \Else
    {%
        Increase~$\delta$.
    }
}

Output $S_\delta$ with the smallest~$\delta$ for which $|S_\delta| \leq |T_r|$.
\label{line:outSdelta}
\end{algorithm}

\begin{theorem}
    \label{theo:rDeltaConDom}
For a given graph~\( G \), Algorithm~\ref{algo:con2DeltaDom} computes a connected \( (r + 2 \Delta) \)-dominating set~\( D \) with \( |D| \leq |D_r| \) in \( \calO \big( m \, \alpha(n) \log \Delta \big) \) time.
\end{theorem}

\begin{proof}
Clearly, the set~$D$ is connected because $D = S_\delta$ for some~$\delta$ and, by Lemma~\ref{lem:SdeltaCardinality}, the set~$S_\delta$ is connected.
By Corollary~\ref{cor:SdeltaCardinality}, for each $\delta \geq \Delta$, $|S_\delta| \leq |T_r|$.
Thus, for each $\delta \geq \Delta$, the binary search decreases~$\delta$ and, eventually, finds some~$\delta$ such that $\delta \leq \Delta$ and $|S_\delta| \leq |T_r|$.
Therefore, the algorithm finds a set~$D$ with $|D| \leq |D_r|$.
Note that, because $D = S_\delta$ for some $\delta \leq \Delta$ and because $S_\delta$ intersects each cluster of~$T_\delta$ (Lemma~\ref{lem:SdeltaCardinality}), it follows from Lemma~\ref{lem:TdeltaDist} that, for each vertex~$v$ of~$G$, $d_\calT(v, D) \leq r(v) + \Delta$ and, by Lemma~\ref{lem:LayPartVertDist}, $d_G(v, D) \leq r(v) + 2 \Delta$.
Thus, $D$ is an $(r + 2 \Delta)$-dominating set for~$G$.

Creating a layering partition for a given graph and computing a minimum connected $r$-dominating set of a tree can be done in linear time~\cite{Dragan1993}.
The one-sided binary search over~$\delta$ has at most $\calO(\log \Delta)$ iterations.
Each iteration of the binary search requires at most linear time to compute~$T_\delta$, $\calO \big( m \, \alpha(n) \big)$ time to compute~$S_\delta$ (Lemma~\ref{lem:SdeltaCardinality}), and constant time to decide whether to increase or decrease~$\delta$.
Therefore, Algorithm~\ref{algo:con2DeltaDom} runs in $\calO \big( m \, \alpha(n) \log \Delta \big)$ total time.
\qed
\end{proof}

\section{Using a Tree-Decomposition}

Theorem~\ref{theo:rDeltaDom} and Theorem~\ref{theo:rDeltaConDom} respectively show how to compute an $(r + \Delta)$-dominating set in linear time and a connected $(r + 2 \Delta)$-dominating set in $\calO \big( m \, \alpha(n) \log \Delta \big)$ time.
It is known that the maximum diameter~$\Delta$ of clusters of any layering partition of a graph approximates the tree-breadth and tree-length of this graph.
Indeed, for a graph~$G$ with $\tl(G) = \lambda$, $\Delta \leq 3 \lambda$~\cite{DouDraGavYan2007}.

\begin{corollary}
Let \( D \) be a minimum \( r \)-dominating set for a given graph~\( G \) with \( \tl(G) = \lambda \).
An \( (r + 3 \lambda) \)-dominating set~\( D' \) for~\( G \) with \( |D'| \leq |D| \) can be computed in linear time.
\end{corollary}

\begin{corollary}
Let \( D \) be a minimum connected \( r \)-dominating set for a given graph~\( G \) with \( \tl(G) = \lambda \).
A connected \( (r + 6 \lambda) \)-dominating set~\( D' \) for~\( G \) with \( |D'| \leq |D| \) can be computed in \( \calO \big( m \, \alpha(n) \log \lambda \big) \) time.
\end{corollary}

In this section, we consider the case when we are given a tree-decomposition with breadth~$\rho$ and length~$\lambda$.
We present algorithms to compute an $(r + \rho)$-dominating set as well as a connected $\big(r + \min (3 \lambda, 5 \rho) \big)$-dominating set in $\calO(nm)$~time.

For the remainder of this section, assume that we are given a graph~$G = (V, E)$ and a tree-decomposition~$\calT$ of~$G$ with breadth~$\rho$ and length~$\lambda$.
We assume that $\rho$ and~$\lambda$ are known and that, for each bag~$B$ of~$\calT$, we know a vertex~$c(B)$ with $B \subseteq N_G^\rho[c(B)]$.
Let $\calT$ be minimal, \ie, $B \nsubseteq B'$ for any two bags $B$ and~$B'$.
Thus, the number of bags is not exceeding the number vertices of~$G$.
Additionally, let each vertex of~$G$ store a list of bags containing it and let each bag of~$T$ store a list of vertices it contains.
One can see this as a bipartite graph where one subset of vertices are the vertices of~$G$ and the other subset are the bags of~$\calT$.
Therefore, the total input size is in $\calO(n + m + M)$ where $M \leq n^2$ is the sum of the cardinality of all bags of~$\calT$.

\subsection{Preprocessing}

Before approaching the (Connected) $r$-Domination problem, we compute a subtree~$\calT'$ of~$\calT$ such that, for each vertex~$v$ of~$G$, $\calT'$ contains a bag~$B$ with $d_G(v, B) \leq r(v)$.
We call such a (not necessarily minimal) subtree an \emph{\( r \)-covering subtree of~\( \calT \)}.

Let $T_r$ be a minimum $r$-covering subtree of~$\calT$.
We do not know how to compute $T_r$ directly.
However, if we are given a bag~$B$ of~$\calT$, we can compute the smallest $r$-covering subtree~$T_B$ which contains~$B$.
Then, we can identify a bag~$B'$ in~$T_B$ for which we know it is a bag of~$T_r$.
Thus, we can compute $T_r$ by computing the smallest $r$-covering subtree which contains~$B'$.

The idea for computing~$T_B$ is to determine, for each vertex~$v$ of~$G$, the bag~$B_v$ of~$\calT$ for which $d_G(v, B_v) \leq r(v)$ and which is closet to~$B$.
Then, let $T_B$ be the smallest tree that contains all these bags~$B_v$.
Algorithm~\ref{algo:T_rWithBag} below implements this approach.

Additionally to computing the tree~$T_B$, we make it a rooted tree with $B$ as the root, give each vertex~$v$ a pointer~$\beta(v)$ to a bag of~$T_B$, and give each bag~$B'$ a counter~$\sigma(B')$.
The pointer~$\beta(v)$ identifies the bag~$B_v$ which is closest to~$B$ in~$T_B$ and intersects the $r$-neighbourhood of~$v$.
The counter~$\sigma(B')$ states the number of vertices~$v$ with $\beta(v) = B'$.
Even though setting $\beta$ and~$\sigma$ as well as rooting the tree are not necessary for computing~$T_B$, we use it when computing an $(r + \rho)$-dominating set later.

\begin{algorithm}
    [htb]
    \caption
    {%
        Computes the smallest $r$-covering subtree~$T_B$ of a given tree-decomposition~$\calT$ that contains a given bag~$B$ of~$\calT$.
    }
    \label{algo:T_rWithBag}

Make $\calT$ a rooted tree with the bag~$B$ as the root.

Create a set $\calB$ of bags and initialise it with $\calB := \{ B \}$.

For each bag~$B'$ of~$\calT$, set $\sigma(B') := 0$ and determine $d_\calT(B', B)$.
\label{line:bagDist}

For each vertex~$u$, determine the bag~$B(u)$ which contains~$u$ and has minimal distance to~$B$.
\label{line:determineBWithU}

\ForEach
{%
    \( u \in V \)%
    \label{line:loopBv}
}
{%
    Determine a vertex~$v$ such that $d_G(u, v) \leq r(u)$ and $d_\calT \big( B(v), B \big)$ is minimal and let $B_u := B(v)$.
    \label{line:determineBu}

    Add $B_u$ to $\calB$, set $\beta(u) := B_u$, and increase $\sigma(B_u)$ by~$1$.
    \label{line:addBvToB}
}

Output the smallest subtree~$T_B$ of~$\calT$ that contains all bags in~$\calB$.
\label{line:ConstructTB}
\end{algorithm}

\begin{lemma}
    \label{lem:TBAlgo}
For a given tree-decomposition~\( \calT \) and a given bag~\( B \) of~\( \calT \), Algorithm~\ref{algo:T_rWithBag} computes an \( r \)-covering subtree~\( T_B \) in \( \calO(nm) \) time such that \( T_B \) contains~\( B \) and has a minimal number of bags.
\end{lemma}

\begin{proof}
    [Correctness]
Note that, by construction of the set~$\calB$ (line~\ref{line:loopBv} to line~\ref{line:addBvToB}), $\calB$ contains a bag~$B_u$ for each vertex~$u$ of~$G$ such that $d_G(u, B_u) \leq r(u)$.
Thus, each subtree of~$\calT$ which contains all bags of~$\calB$ is an $r$-covering subtree.
To show the correctness of the algorithm, it remains to show that the smallest $r$-covering subtree of~$\calT$ which contains~$B$ has to contain each bag from the set~$\calB$.
Then, the subtree~$T_B$ constructed in line~\ref{line:ConstructTB} is the desired subtree.

By properties of tree-decompositions, the set of bags which intersect the $r$-neighbourhood of some vertex~$u$ induces a subtree~$T_u$ of~$\calT$.
That is, $T_u$ contains exactly the bags~$B'$ with $d_G(u, B') \leq r(u)$.
Note that $\calT$ is a rooted tree with $B$ as the root.
Clearly, the bag~$B_u \in \calB$ (determined in line~\ref{line:determineBu}) is the root of~$T_u$ since it is the bag closest to~$B$.
Hence, each bag~$B'$ with $d_G(u, B') \leq r(u)$ is a descendant of~$B_u$.
Therefore, if a subtree of~$\calT$ contains~$B$ and does not contain~$B_u$, then it also cannot contain any descendant of~$B_u$ and, thus, contains no bag intersecting the $r$-neighbourhood of~$u$.
\qed
\end{proof}

\begin{proof}
    [Complexity]
Recall that $\calT$ has at most $n$~bags and that the sum of the cardinality of all bags of~$\calT$ is~$M \leq n^2$.
Thus, line~\ref{line:bagDist} and line~\ref{line:determineBWithU} require at most $\calO(M)$~time.
Using a BFS, it takes at most $\calO(m)$~time, for a given vertex~$u$, to determine a vertex~$v$ such that $d_G(u, v) \leq r(u)$ and $d_\calT \big( B(v), B \big)$ is minimal (line~\ref{line:determineBu}).
Therefore, the loop starting in line~\ref{line:loopBv} and, thus, Algorithm~\ref{algo:T_rWithBag} run in at most $\calO(nm)$ total time.
\qed
\end{proof}

Lemma~\ref{lem:rDiskLeaf} and Lemma~\ref{lem:TBleavesInTr} below show that each leaf~$B' \neq B$ of~$T_B$ is a bag of a minimum $r$-covering subtree~$T_r$ of~$\calT$.
Note that both lemmas only apply if $T_B$ has at least two bags.
If $T_B$ contains only one bag, it is clearly a minimum $r$-covering subtree.

\begin{lemma}
    \label{lem:rDiskLeaf}
For each leaf~\( B' \neq B \) of~\( T_B \), there is a vertex~\( v \) in~\( G \) such that \( B' \) is the only bag of~\( T_B \) with \( d_G(v, B') \leq r(v) \).
\end{lemma}

\begin{proof}
Assume that Lemma~\ref{lem:rDiskLeaf} is false.
Then, there is a leaf~$B'$ such that, for each vertex~$v$ with $d_G(v, B') \leq r(v)$, $T_B$ contains a bag~$B'' \neq B'$ with $d_G(v, B'') \leq r(v)$.
Thus, for each vertex~$v$, the $r$-neighbourhood of~$v$ is intersected by a bag of the tree-decomposition~$T_B - B'$.
This contradicts with the minimality of~$T_B$.
\qed
\end{proof}

\begin{lemma}
    \label{lem:TBleavesInTr}
For each leaf~\( B' \neq B \) of~\( T_B \), there is a minimum $r$-covering subtree~\( T_r \) of~\( \calT \) which contains~\( B'\).
\end{lemma}

\begin{proof}
Assume that $T_r$ is a minimum $r$-covering subtree which does not contain~$B'$.
Because of Lemma~\ref{lem:rDiskLeaf}, there is a vertex~$v$ of~$G$ such that $B'$ is the only bag of~$T_B$ which intersects the $r$-neighbourhood of~$v$.
Therefore, $T_r$ contains only bags which are descendants of~$B'$.
Partition the vertices of~$G$ into the sets $V^\uparrow$ and~$V^\downarrow$ such that $V^\downarrow$ contains the vertices of~$G$ which are contained in~$B'$ or in a descendant of~$B'$.
Because $T_r$ is an $r$-covering subtree and because $T_r$ only contains descendants of~$B'$, it follows from properties of tree-decompositions that, for each vertex~$v \in V^\uparrow$, there is a path of length at most~$r(v)$ from~$v$ to a bag of~$T_r$ passing through~$B'$ and, thus, $d_G(v, B') \leq r(v)$.
Similarly, since $T_B$ is an $r$-covering subtree, it follows that, for each vertex~$v \in V^\downarrow$, $d_G(v, B') \leq r(v)$.
Therefore, for each vertex~$v$ of~$G$, $d_G(v, B') \leq r(v)$ and, thus, $B'$ induces an $r$-covering subtree~$T_r$ of~$\calT$ with $|T_r| = 1$.
\qed
\end{proof}

Algorithm~\ref{algo:T_rDeco} below uses Lemma~\ref{lem:TBleavesInTr} to compute a minimum $r$-covering subtree~$T_r$ of~$\calT$.

\begin{algorithm}
    [htb]
    \caption
    {%
        Computes a minimum $r$-covering subtree~$T_r$ of a given tree-decomposition~$\calT$.
    }
    \label{algo:T_rDeco}

Pick an arbitrary bag~$B$ of~$\calT$.

Determine the subtree~$T_B$ of~$\calT$ using Algorithm~\ref{algo:T_rWithBag}.

\If
{%
    \( |T_B| = 1 \)
}
{%
    Output~$T_r := T_B$.
}
\Else
{%
    Select an arbitrary leaf~$B' \neq B$ of~$T_B$.

    Determine the subtree~$T_{B'}$ of~$\calT$ using Algorithm~\ref{algo:T_rWithBag}.

    Output~$T_r := T_{B'}$.
}
\end{algorithm}

\begin{lemma}
    \label{lem:computeTr}
Algorithm~\ref{algo:T_rDeco} computes a minimum \( r \)-covering subtree~\( T_r \) of~\( \calT \) in \( \calO(nm) \) time.
\end{lemma}

\begin{proof}
Algorithm~\ref{algo:T_rDeco} first picks an arbitrary bag~$B$ and then uses Algorithm~\ref{algo:T_rWithBag} to compute the smallest $r$-covering subtree~$T_B$ of~$\calT$ which contains~$B$.
By Lemma~\ref{lem:TBleavesInTr}, for each leaf~$B'$ of~$T_B$, there is a minimum $r$-covering subtree~$T_r$ which contains~$B'$.
Thus, performing Algorithm~\ref{algo:T_rWithBag} again with $B'$ as input creates such a subtree~$T_r$.

Clearly, with exception of calling Algorithm~\ref{algo:T_rWithBag}, all steps of Algorithm~\ref{algo:T_rDeco} require only constant time.
Because Algorithm~\ref{algo:T_rWithBag} requires at most $\calO(nm)$ time (see Lemma~\ref{lem:TBAlgo}) and is called at most two times, Algorithm~\ref{algo:T_rDeco} runs in at most $\calO(nm)$ total time.
\qed
\end{proof}

Algorithm~\ref{algo:T_rDeco} computes~$T_r$ by, first, computing~$T_B$ for some bag~$B$ and, second, computing~$T_{B'} = T_r$ for some leaf~$B'$ of~$T_B$.
Note that, because both trees are computed using Algorithm~\ref{algo:T_rWithBag}, Lemma~\ref{lem:rDiskLeaf} applies to $T_B$ and~$T_{B'}$.
Therefore, we can slightly generalise Lemma~\ref{lem:rDiskLeaf} as follows.

\begin{corollary}
    \label{cor:rDiskLeaf}
For each leaf~\( B \) of~\( T_r \), there is a vertex~\( v \) in~\( G \) such that \( B \) is the only bag of~\( T_r \) with \( d_G(v, B) \leq r(v) \).
\end{corollary}

\subsection{$r$-Domination}

In this subsection, we use the minimum $r$-covering subtree~$T_r$ to determine an $(r + \rho)$-dominating set~$S$ in $\calO(nm)$ time using the following approach.
First, compute~$T_r$.
Second, pick a leaf~$B$ of~$T_r$.
If there is a vertex~$v$ such that $v$ is not dominated and $B$ is the only bag intersecting the $r$-neighbourhood of~$v$, then add the center of~$B$ into~$S$, flag all vertices~$u$ with $d_G(u, B) \leq r(u)$ as dominated, and remove $B$ from~$T_r$.
Repeat the second step until $T_r$ contains no more bags and each vertex is flagged as dominated.
Algorithm~\ref{algo:rhoDomTreeDeco} below implements this approach.
Note that, instead of removing bags from~$T_r$, we use a reversed BFS-order of the bags to ensure the algorithm processes bags in the correct order.

\begin{algorithm}
    [htb]
    \caption
    {%
        Computes an $(r + \rho)$-dominating set~$S$ for a given graph~$G$ with a given tree-decomposition~$\calT$ with breadth~$\rho$.
    }
    \label{algo:rhoDomTreeDeco}

Compute a minimum $r$-covering subtree~$T_r$ of~$\calT$ using Algorithm~\ref{algo:T_rDeco}.
\label{line:computeTr}

Give each vertex~$v$ a binary flag indicating if $v$ is \emph{dominated}.
Initially, no vertex is dominated.

Create an empty vertex set~$S_0$.

Let $\langle B_1, B_2, \ldots, B_k \rangle$ be the reverse of a BFS-order of~$T_r$ starting at its root.
\label{line:computeBFSofTr}

\For
{%
    \( i = 1 \) \KwTo \( k \)
}
{%
    \If
    {%
        \( \sigma(B_i) > 0 \)
    }
    {%
        Determine all vertices~$u$ such that $u$ has not been flagged as dominated and that $d_G(u, B_i) \leq r(u)$.
        Add all these vertices into a new set~$X_i$.
        \label{line:determineXi}

        Let $S_i = S_{i - 1} \cup \big \{ c(B_i) \big \}$.

        For each vertex~$u \in X_i$, flag $u$ as dominated, and decrease $\sigma \big( \beta(u) \big)$ by~$1$.
        \label{line:flagVertex}
    }
    \Else
    {%
        Let $S_i = S_{i - 1}$.
        \label{line:setSi}
    }
}
Output $S := S_k$.
\end{algorithm}

\begin{theorem}
Let \( D \) be a minimum \( r \)-dominating set for a given graph~\( G \).
Given a tree-decomposition with breadth~\( \rho \) for~\( G \), Algorithm~\ref{algo:rhoDomTreeDeco} computes an \( (r + \rho) \)-dominating set~\( S \) with \( |S| \leq |D| \) in \( \calO(nm) \) time.
\end{theorem}

\begin{proof}
    [Correctness]
First, we show that $S$ is an $(r + \rho)$-dominating set for~$G$.
Note that a vertex~$v$ is flagged as dominated only if $S_i$ contains a vertex~$c(B_j)$ with $d_G(v, B_j) \leq r(v)$ (see line~\ref{line:determineXi} to line~\ref{line:flagVertex}).
Thus, $v$ is flagged as dominated only if $d_G(v, S_i) \leq d_G \big(v, c(B_j) \big) \leq r(v) + \rho$.
Additionally, by construction of~$T_r$ (see Algorithm~\ref{algo:T_rWithBag}), for each vertex~$v$, $T_r$~contains a bag~$B$ with $\beta(v) = B$, $\sigma(B)$ states the number of vertices~$v$ with $\beta(v) = B$, and $\sigma(B)$ is decreased by~$1$ only if such a vertex~$v$ is flagged as dominated (see line~\ref{line:flagVertex}).
Therefore, if $G$ contains a vertex~$v$ with $d_G(v, S_i) > r(v) + \rho$, then $v$ is not flagged as dominated and $T_r$ contains a bag~$B_i$ with $\beta(v) = B_i$ and $\sigma(B_i) > 0$.
Thus, when $B_i$ is processed by the algorithm, $c(B_i)$ will be added to~$S_i$ and, hence, $d_G(v, S_i) \leq r(v) + \rho$.

Let $V^S_i = \{ \, u \mid d_G(u, B_j) \leq r(u), c(B_j) \in S_i \, \}$ be the set of vertices which are flagged as dominated after the algorithm processed~$B_i$, \ie, each vertex in $V^S_i$ is $(r + \rho)$-dominated by~$S_i$.
Similarly, for some set~$D_i \subseteq D$, let $V^D_i = \{ \, u \mid d_G(u, D_i) \leq r(u) \, \}$ be the set of vertices dominated by~$D_i$.
To show that $|S| \leq |D|$, we show by induction over~$i$ that, for each~$i$, (i)~there is a set~$D_i \subseteq D$ such that $V^D_i \subseteq V^S_i$, (ii)~$|S_i| = |D_i|$, and (iii)~if, for some vertex~$v$, $\beta(v) = B_j$ with $j \leq i$, then $v \in V^S_i$.

For the base case, let $S_0 = D_0 = \emptyset$.
Then, $V^S_0 = V^D_0 = \emptyset$ and all three statements are satisfied.
For the inductive step, first, consider the case when $\sigma(B_i) = 0$.
Because $\sigma(B_i) = 0$, each vertex~$v$ with $\beta(v) = B_i$ is flagged as dominated, \ie, $v \in V^S_{i-1}$.
Thus, by setting $S_i = S_{i - 1}$ (line~\ref{line:setSi}) and $D_i = D_{i - 1}$, all three statements are satisfied for~$i$.
Next, consider the case when $\sigma(B_i) > 0$.
Therefore, $G$ contains a vertex~$u$ with $\beta(u) = B_i$ and $u \notin V^S_{i-1}$.
Then, the algorithm sets $S_i = S_{i-1} \cup \big \{ c(B_i) \big \}$ and flags all such $u$ as dominated (see line~\ref{line:determineXi} to line~\ref{line:flagVertex}).
Thus, $u \in V^S_i$ and statement~(iii) is satisfied.
Let $d_u$ be a vertex in $D$ with minimal distance to~$u$.
Thus, $d_G(d_u, u) \leq r(u)$, \ie, $d_u$ is in the $r$-neighbourhood of~$u$.
Note that, because $u \notin V^S_{i-1}$ and $V^D_{i-1} \subseteq V^S_{i-1}$, $d_u \notin D_{i-1}$.
Therefore, by setting $D_i = D_{i-1} \cup \{ d_u \}$, $|S_i| = |S_{i-1}| + 1 = |D_{i-1}| + 1 = |D_i|$ and statement~(ii) is satisfied.
Recall that $\beta(u)$ points to the bag closest to the root of~$T_r$ which intersects the $r$-neighbourhood of~$u$.
Thus, because $\beta(u) = B_i$, each bag~$B \neq B_i$ with $d_G(u, B) \leq r(u)$ is a descendant of~$B_i$.
Therefore, $d_u$ is in~$B_i$ or in a descendant of~$B_i$.
Let $v$ be an arbitrary vertex of~$G$ such that $v \notin V^S_{i-1}$ and $d_G(v, d_u) \leq r(v)$, \ie, $v$ is dominated by~$d_u$ but not by~$S_{i-1}$.
Due to statement~(iii) of the induction hypothesis, $\beta(v) = B_j$ with $j \geq i$, \ie, $B_j$ cannot be a descendant of~$B_i$.
Partition the vertices of~$G$ into the sets $V^\uparrow_i$ and~$V^\downarrow_i$ such that $V^\downarrow_i$ contains the vertices which are contained in~$B_i$ or in a descendant of~$B_i$.
If $v \in V^\downarrow_i$, then there is a path of length at most~$r(v)$ from $v$ to~$B_j$ passing through~$B_i$.
If $v \in V^\uparrow_i$, then, because $d_u \in V^\downarrow_i$, there is a path of length at most~$r(v)$ from $v$ to~$d_u$ passing through~$B_i$.
Therefore, $d_G(v, B_i) \leq r(v)$.
That is, each vertex $r$-dominated by~$d_u$, is $(r + \rho)$-dominated by some~$c(B_j) \in S_i$.
Therefore, because $S_i = S_{i-1} \cup \big \{ c(B_i) \big \}$ and $D_i = D_{i-1} \cup \{ d_u \}$, $v \in V^S_i \cap V^D_i $ and, thus, statement~(i) is satisfied.
\qed
\end{proof}

\begin{proof}
    [Complexity]
Computing~$T_r$ (line~\ref{line:computeTr}) takes at most~$\calO(nm)$ time (see Lemma~\ref{lem:computeTr}).
Because $T_r$ has at most $n$ bags, computing a BFS-order of $T_r$ (line~\ref{line:computeBFSofTr}) takes at most $\calO(n)$ time.
For some bag~$B_i$, determining all vertices~$u$ with $d_G(u, B_i) \leq r(u)$, flagging $u$ as dominated, and decreasing $\sigma \big( \beta(u) \big)$ (line~\ref{line:determineXi} to line~\ref{line:flagVertex}) can be done in $\calO(m)$ time by performing a BFS starting at all vertices of~$B_i$ simultaneously.
Therefore, because $T_r$ has at most $n$ bags, Algorithm~\ref{algo:rhoDomTreeDeco} requires at most $\calO(nm)$ total time.
\qed
\end{proof}

\subsection{Connected $r$-Domination}

In this subsection, we show how to compute a connected $(r + 5 \rho)$-dominating set and a connected $(r + 3 \lambda)$-dominating set for~$G$.
For both results, we use almost the same algorithm.
To identify and emphasise the differences, we use the label~\rHrt for parts which are only relevant to determine a connected $(r + 5 \rho)$-dominating set and use the label~\rDmd for parts which are only relevant to determine a connected $(r + 3 \lambda)$-dominating set.

For the remainder of this subsection, let $D_r$ be a minimum connected $r$-dominating set of~$G$.
For \rHrt~$\phi = 3 \rho$ or \rDmd~$\phi = 2 \lambda$, let $T_\phi$ be a minimum $(r + \phi)$-covering subtree of~$\calT$ as computed by Algorithm~\ref{algo:T_rDeco}.

The idea of our algorithm is to, first, compute~$T_\phi$ and, second, compute a small enough connected set~$C_\phi$ such that $C_\phi$ intersects each bag of~$T_\phi$.
Lemma~\ref{lem:rPhiDomSet} below shows that such a set~$C_\phi$ is an $\big( r + (\phi + \lambda) \big)$-dominating set.

\begin{lemma}
    \label{lem:rPhiDomSet}
Let \( C_\phi \) be a connected set that contains at least one vertex of each leaf of~\( T_\phi \).
Then, \( C_\phi \) is an \( \big( r + (\phi + \lambda) \big) \)-dominating set.
\end{lemma}

\begin{proof}
Clearly, since $C_\phi$ is connected and contains a vertex of each leaf of~$T_\phi$, $C_\phi$ contains a vertex of every bag of~$T_\phi$.
By construction of~$T_\phi$, for each vertex~$v$ of~$G$, $T_\phi$ contains a bag~$B$ such that $d_G(v, B) \leq r(v) + \phi$.
Therefore, $d_G(v, C_\phi) \leq r(v) + \phi + \lambda$, \ie, $C_\phi$ is an $\big( r + (\phi + \lambda) \big)$-dominating set.
\qed
\end{proof}

To compute a connected set~$C_\phi$ which intersects all leaves of~$T_\phi$, we first consider the case when $T_\rho$ contains only one bag~$B$.
In this case, we can construct~$C_\phi$ by simply picking an arbitrary vertex~$v \in B$ and setting $C_\phi = \{ v \}$.
Similarly, if $T_\rho$ contains exactly two bags $B$ and~$B'$, pick a vertex~$v \in B \cap B'$ and set $C_\phi = \{ v \}$.
In both cases, due to Lemma~\ref{lem:rPhiDomSet}, $C_\phi$ is clearly an $\big( r + (\phi + \lambda) \big)$-dominating set with $|C_\phi| \leq |D_r|$.

Now, consider the case when $T_\phi$ contains at least three bags.
Additionally, assume that $T_\phi$ is a rooted tree such that its root~$R$ is a leaf.

\subsubsection{Notation.}
Based on its degree in~$T_\phi$, we refer to each bag~$B$ of~$T_\phi$ either as leaf, as \emph{path bag} if $B$ has degree~$2$, or as \emph{branching bag} if $B$ has a degree larger than~$2$.
Additionally, we call a maximal connected set of path bags a \emph{path segment} of~$T_\phi$.
Let $\bbL$ denote the set of leaves, $\bbP$ denote the set of path segments, and $\bbB$ denote the set of branching bags of~$T_\phi$.
Clearly, for any given tree~$T$, the sets $\bbL$, $\bbP$, and~$\bbB$ are pairwise disjoint and can be computed in linear time.

Let $B$ and~$B'$ be two adjacent bags of~$T_\phi$ such that $B$ is the parent of~$B'$.
We call $S = B \cap B'$ the \emph{up-separator} of~$B'$, denoted as~$\Sup(B')$, and a \emph{down-separator} of~$B$, denoted as $\Sdown(B)$, \ie, $S = \Sup(B') = \Sdown(B)$.
Note that a branching bag has multiple down-separators and that (with exception of~$R$) each bag has exactly one up-separator.
For each branching bag~$B$, let~$\calSdown(B)$ be the set of down-separators of~$B$.
Accordingly, for a path segment~$P \in \bbP$, $\Sup(P)$ is the up-separator of the bag in~$P$ closest to the root and $\Sdown(P)$ is the down separator of the bag in~$P$ furthest from the root.
Let $\nu$ be a function that assigns a vertex of~$G$ to a given separator.
Initially, $\nu(S)$ is undefined for each separator~$S$.

\subsubsection{Algorithm.}
Now, we show how to compute~$C_\phi$.
We, first, split $T_\phi$ into the sets $\bbL$, $\bbP$, and~$\bbB$.
Second, for each~$P \in \bbP$, we create a small connected set~$C_P$, and, third, for each~$B \in \bbB$, we create a small connected set~$C_B$.
If this is done properly, the union~$C_\phi$ of all these sets forms a connect set which intersects each bag of~$T_\phi$.

Note that, due to properties of tree-decompositions, it can be the case that there are two bags $B$ and~$B'$ which have a common vertex~$v$, even if $B$ and~$B'$ are non-adjacent in~$T_\phi$.
In such a case, either $v \in \Sdown(B) \cap \Sup(B')$ if $B$ is an ancestor of~$B'$, or $v \in \Sup(B) \cap \Sup(B')$ if neither is ancestor of the other.
To avoid problems caused by this phenomena and to avoid counting vertices multiple times, we consider any vertex in an up-separator as part of the bag above.
That is, whenever we process some segment or bag~$X \in \bbL \cup \bbP \cup \bbB$, even though we add a vertex~$v \in \Sup(X)$ to~$C_\phi$, $v$ is not contained in~$C_X$.

\paragraph{Processing Path Segments.}
First, after splitting~$T_\phi$, we create a set~$C_{P}$ for each path segment~$P \in \bbP$ as follows.
We determine $\Sup(P)$ and~$\Sdown(P)$ and then find a shortest path~$Q_P$ from $\Sup(P)$ to~$\Sdown(P)$.
Note that $Q_P$ contains exactly one vertex from each separator.
Let $x \in \Sup(P)$ and $y \in \Sdown(P)$ be these vertices.
Then, we set $\nu \big( \Sup(P) \big) = x$ and $\nu \big( \Sdown(P) \big) = y$.
Last, we add the vertices of~$Q_P$ into~$C_\phi$ and define $C_P$ as $Q_P \setminus \Sup(P)$.
Let $C_\bbP$ be the union of all sets~$C_P$, \ie, $C_\bbP = \bigcup_{P \in \bbP} C_P$.

\begin{lemma}
    \label{lem:CPcardinality}
\( |C_\bbP| \leq |D_r| - \phi \cdot \Lambda \big( T_\phi \big) \).
\end{lemma}

\begin{proof}
Recall that $T_\phi$ is a minimum $(r + \phi)$-covering subtree of~$\calT$.
Thus, by Corollary~\ref{cor:rDiskLeaf}, for each leaf~$B \in \bbL$ of~$T_\phi$, there is a vertex~$v$ in~$G$ such that $B$ is the only bag of~$T_\phi$ with $d_G(v, B) \leq r(v) + \phi$.
Because $D_r$ is a connected $r$-dominating set, $D_r$ intersects the $r$-neighbourhood of each of these vertices~$v$.
Thus, by properties of tree-decompositions, $D_r$ intersects each bag of~$T_\phi$.
Additionally, for each such~$v$, $D_r$ contains a path~$D_v$ with $|D_v| \geq \phi$ such that $D_v$ intersects the $r$-neighbourhood of~$v$, intersects the corresponding leaf~$B$ of~$T_\phi$, and does not intersect $\Sup(B)$ ($\Sdown(B)$ if $B = R$).
Let $D_\bbL$ be the union of all such sets~$D_v$.
Therefore, $|D_\bbL| \geq \phi \cdot \Lambda \big( T_\phi \big)$.

Because $D_r$ intersects each bag of~$T_\phi$, $D_r$ also intersects the up- and down-separators of each path segment.
For a path segment~$P \in \bbP$, let $x$ and~$y$ be two vertices of~$D_r$ such that $x \in \Sup(P)$, $y \in \Sdown(P)$, and for which the distance in~$G[D_r]$ is minimal.
Let $D_P$ be the set of vertices on the shortest path in~$G[D_r]$ from $x$ to~$y$ without~$x$, \ie, $x \notin D_P$.
Note that, by construction, for each~$P \in \bbP$, $D_P$ contains exactly one vertex in $\Sdown(P)$ and no vertex in~$\Sup(P)$.
Thus, for all~$P, P' \in \bbP$, $D_P \cap D_{P'} = \emptyset$.
Let $D_\bbP$ be the union of all such sets~$D_P$, \ie, $D_\bbP = \bigcup_{P \in \bbP} D_P$.
By construction, $|D_\bbP| = \sum_{P \in \bbP} |D_P|$ and $D_\bbL \cap D_\bbP = \emptyset$.
Therefore, $|D_r| \geq |D_\bbP| + |D_\bbL|$ and, hence,
\[
    \sum_{P \in \bbP} |D_P| \leq |D_r| - |D_\bbL| \leq |D_r| - \phi \cdot \Lambda \big( T_\phi \big).
\]

Recall that, for each $P \in \bbP$, the sets $C_P$ and~$D_P$ are constructed based on a path from $\Sup(P)$ to~$\Sdown(P)$.
Since $C_P$ is based on a shortest path in~$G$, it follows that $|C_P| = d_G \big( \Sup(P), \Sdown(P) \big) \leq |D_P|$.
Therefore,
\[
    |C_\bbP| \leq \sum_{P \in \bbP} |C_P| \leq \sum_{P \in \bbP} |D_P| \leq |D_r| - \phi \cdot \Lambda \big( T_\phi \big).
    \tag*{\qed}
\]
\end{proof}

\paragraph{Processing Branching Bags.}
After processing path segments, we process the branching bags of~$T_\phi$.
Similar to path segments, we have to ensure that all separators are connected.
Branching bags, however, have multiple down-separators.
To connect all separators of some bag~$B$, we pick a vertex~$s$ in each separator~$S \in \calSdown(B) \cup \big \{ \Sup(B) \big \}$.
If $\nu(S)$ is defined, we set $s = \nu(S)$.
Otherwise, we pick an arbitrary~$s \in S$ and set $\nu(S) = s$.
Let $\calSdown(B) = \{ S_1, S_2, \ldots \}$, $s_i = \nu(S_i)$, and $t = \nu \big( \Sup(B) \big)$.
We then connect these vertices as follows.
(See Figure~\ref{fig:branchingBagConstr} for an illustration.)
\begin{enumerate}[\rHrt]
    \item[\rHrt]
        Connect each vertex~$s_i$ via a shortest path~$Q_i$ (of length at most~$\rho$) with the center~$c(B)$ of~$B$.
        Additionally, connect~$c(B)$ via a shortest path~$Q_t$ (of length at most~$\rho$) with~$t$.
        Add all vertices from the paths~$Q_i$ and from the path~$Q_t$ into $C_\phi$ and let $C_B$ be the union of these paths without~$t$.
    \item[\rDmd]
        Connect each vertex~$s_i$ via a shortest path~$Q_i$ (of length at most~$\lambda$) with~$t$.
        Add all vertices from the paths~$Q_i$ into $C_\phi$ and let $C_B$ be the union of these paths without~$t$.
\end{enumerate}
Let $C_\bbB$ be the union of all created sets~$C_B$, \ie, $C_\bbB = \bigcup_{B \in \bbB} C_B$.

\begin{figure}
    [htb]
    \centering

    \begin{tabular}{c}
        \includegraphics[]{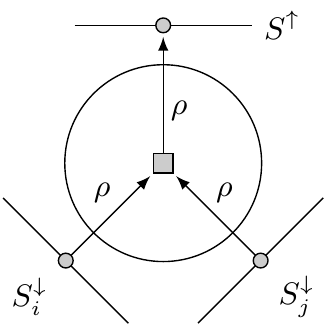} \\[5pt]
        \rHrt
    \end{tabular}
    \hfil
    \begin{tabular}{c}
        \includegraphics[]{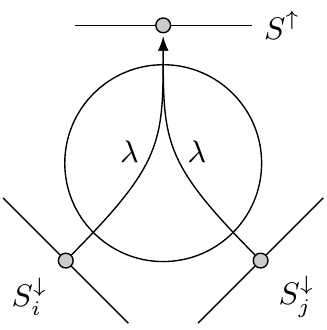} \\[5pt]
        \rDmd
    \end{tabular}

    \caption
    {%
        Construction of the set~$C_B$ for a branching bag~$B$.
    }
    \label{fig:branchingBagConstr}
\end{figure}

Before analysing the cardinality of~$C_\bbB$ in Lemma~\ref{lem:CBcardinality} below, we need an axillary lemma.

\begin{lemma}
    \label{lem:treeLeafsBranch}
For a tree~\( T \) which is rooted in one of its leaves, let \( b \) denote the number of branching nodes, \( c \) denote the total number of children of branching nodes, and \( l \) denote the number of leaves.
Then, \( c + b \leq 3 l - 1 \) and \( c \leq 2 l - 1 \).
\end{lemma}

\begin{proof}
Assume that we construct $T$ by starting with only the root and then step by step adding leaves to it.
Let $T_i$ be the subtree of~$T$ with $i$~nodes during this construction.
We define $b_i$, $c_i$, and~$l_i$ accordingly.
Now, assume by induction over~$i$ that Lemma~\ref{lem:treeLeafsBranch} is true for~$T_i$.
Let $v$ be the leaf we add to construct~$T_{i+1}$ and let $u$ be its neighbour.

First, consider the case when $u$ is a leaf of~$T_i$.
Then, $u$ becomes a path node of~$T_{i+1}$.
Therefore, $b_{i+1} = b_i$, $c_{i+1} = c_i$, and $l_{i+1} = l_i$.
Next, assume that $u$ is path node of~$T_i$.
Then, $u$ is a branch node of~$T_{i+1}$.
Thus, $b_{i+1} = b_i + 1$, $c_{i+1} = c_i + 2$, and $l_{i+1} = l_i + 1$.
Therefore, $c_{i+1} + b_{i+1} = c_i + b_i + 3 \leq 3 (l_i + 1) - 1 = 3 l_{i+1} - 1$ and $c_{i+1} = c_i + 2 \leq 2 (l_i + 1) - 1 = 2 l_{i+1} - 1$.
It remains to check the case when $u$ is a branch node of~$T_i$.
Then, $b_{i+1} = b_i$, $c_{i+1} = c_i + 1$, and $l_{i+1} = l_i + 1$.
Thus, $c_{i+1} + b_{i+1} = c_i + b_i + 1 \leq 3 l_i - 1 + 1 \leq 3 l_{i+1} - 1$ and $c_{i+1} = c_i + 1 \leq 2 l_i - 1 + 1 \leq 2 l_{i+1} - 1$.
Therefore, in all three cases, Lemma~\ref{lem:treeLeafsBranch} is true for~$T_{i+1}$.
\qed
\end{proof}

\begin{lemma}
    \label{lem:CBcardinality}
\( |C_\bbB| \leq \phi \cdot \Lambda \big( T_\phi \big) \).
\end{lemma}

\begin{proof}
For some branching bag~$B \in \bbB$, the set~$C_B$ contains \rHrt~a path of length at most~$\rho$ for each~$S_i \in \calSdown(B)$ and a path of length at most~$\rho$ to~$\Sup(B)$, or \rDmd~a path of length at most~$\lambda$ for each~$S_i \in \calSdown(B)$.
Thus, \rHrt~$|C_B| \leq \rho \cdot \big| \calSdown(B) \big| + \rho$ or \rDmd~$|C_B| \leq \lambda \cdot \big| \calSdown(B) \big|$.
Recall that $\calSdown(B)$ contains exactly one down-separator for each child of~$B$ in~$T_\phi$ and that $C_\bbB$ is the union of all sets~$C_B$.
Therefore, Lemma~\ref{lem:treeLeafsBranch} implies the following.
\begin{alignat*}{2}
    |C_\bbB|
        & \leq \sum_{B \in \bbB} |C_B| \\
    \rHrt \enspace
        &  \leq \rho \cdot \! \sum_{B \in \bbB} \big| \calSdown(B) \big| + \rho \cdot |\bbB|
        && \leq 3 \rho \cdot \Lambda \big( T_\phi \big) - 1 \\
    \rDmd \enspace
        &  \leq \lambda \cdot \! \sum_{B \in \bbB} \big| \calSdown(B) \big|
        && \leq 2 \lambda \cdot \Lambda \big( T_\phi \big) - 1 \\
        &  \leq \phi \cdot \Lambda \big( T_\phi \big) - 1.
        \tag*{\qed}
\end{alignat*}
\end{proof}

\paragraph{Properties of \( C_\phi \).}
We now analyse the created set~$C_\phi$ and show that $C_\phi$ is a connected $(r + \phi)$-dominating set for~$G$.

\begin{lemma}
    \label{lem:CphiBagInt}
\( C_\phi \) contains a vertex in each bag of~\( T_\phi \).
\end{lemma}

\begin{proof}
Clearly, by construction, $C_\phi$ contains a vertex in each path bag and in each branching bag.
Now, consider a leaf~$L$ of~$T_\phi$.
$L$ is adjacent to a path segment or branching bag~$X \in \bbP \cap \bbB$.
Whenever such an $X$ is processed, the algorithm ensures that all separators of~$X$ contain a vertex of~$C_\phi$.
Since one of these separators is also the separator of~$L$, it follows that each leaf~$L$ and, thus, each bag of~$T_\phi$ contains a vertex of~$C_\phi$.
\qed
\end{proof}

\begin{lemma}
    \label{lem:CphiCardinality}
\( |C_\phi| \leq |D_r| \).
\end{lemma}

\begin{proof}
Note that, for each vertex~$u$ we add to~$C_\phi$, we also add $u$ to a unique set~$C_X$ for some $X \in \bbP \cap \bbB$.
The exception is the vertex~$v$ in~$\Sdown(R)$ which is added to no such set~$C_X$.
It follows from our construction of the sets~$C_X$ that there is only one such vertex~$v$ and that $v = \nu \big( \Sdown(R) \big)$.
Thus, $|C_\phi| = |C_\bbP| + |C_\bbB| + 1$.
Now, it follows from Lemma~\ref{lem:CPcardinality} and Lemma~\ref{lem:CBcardinality} that
\[
    |C_\phi| \leq |D_r| - \phi \cdot \Lambda \big( T_\phi \big) + \phi \cdot \Lambda \big( T_\phi \big) - 1 + 1 \leq |D_r|.
    \tag*{\qed}
\]
\end{proof}

\begin{lemma}
    \label{lem:CphiConnected}
\( C_\phi \) is connected.
\end{lemma}

\begin{proof}
First, note that, by maximality, two path segments of~$T_\phi$ cannot share a common separator.
Also, note that, when processing a branching bag~$B$, the algorithm first checks if, for any separator~$S$ of~$B$, $\nu(S)$ is already defined; if this is the case, it will not be overwritten.
Therefore, for each separator~$S$ in~$T_\phi$, $\nu(S)$ is defined and never overwritten.

Next, consider a path segment or branching bag~$X \in \bbP \cup \bbB$ and let $S$ and~$S'$ be two separators of~$X$.
Whenever such an $X$ is processed, our approach ensures that $C_\phi$ connects $\nu(S)$ with~$\nu(S')$.
Additionally, observe that, when processing~$X$, each vertex added to~$C_\phi$ is connected via~$C_\phi$ with~$\nu(S)$ for some separator~$S$ of~$X$.

Thus, for any two separators $S$ and~$S'$ in~$T_\phi$, $C_\phi$ connects $\nu(S)$ with~$\nu(S')$ and, additionally, each vertex~$v \in C_\phi$ is connected via~$C_\phi$ with~$\nu(S)$ for some separator~$S$ in~$T_\phi$.
Therefore, $C_\phi$ is connected.
\qed
\end{proof}

From Lemma~\ref{lem:CphiBagInt}, Lemma~\ref{lem:CphiCardinality}, Lemma~\ref{lem:CphiConnected}, and from applying Lemma~\ref{lem:rPhiDomSet} it follows:

\begin{corollary}
    \label{cor:CphiConDomSet}
\( C_\phi \) is a connected \( \big( r + (\phi + \lambda) \big) \)-dominating set for~\( G \) with \( |C_\phi| \leq |D_r| \).
\end{corollary}

\paragraph{Implementation.}
Algorithm~\ref{algo:conPhiDomination} below implements our approach described above.
This also includes the case when~$T_\phi$ contains at most two bags.

\SetKw{KwStop}{stop}

\begin{algorithm}
    [htb]
    \caption
    {%
        Computes \rHrt~a connected $(r + 5 \rho)$-dominating set or \rDmd~a connected $(r + 3 \lambda)$-dominating set for a given graph~$G$ with a given tree-decomposition~$\calT$ with breadth~$\rho$ and length~$\lambda$.
    }
    \label{algo:conPhiDomination}

\parbox[t]{\hsize}
{%
    \rHrt Set $\phi := 3 \rho$. \\
    \rDmd Set $\phi := 2 \lambda$.
}

Compute a minimum $(r + \phi)$-covering subtree~$T_\phi$ of~$\calT$ using Algorithm~\ref{algo:T_rDeco}.
\label{line:compTphi}

\If
{%
    \( T_\phi \) contains only one bag~\( B \)%
    \label{line:ifTphiOneBag}
}
{%
    Pick an arbitrary vertex~$u \in B$, output~$C_\phi := \{ u \}$, and \KwStop.
}

\If
{%
    \( T_\phi \) contains exactly two bags \( B \) and~\( B' \)
}
{%
    Pick an arbitrary vertex~$u \in B \cap B'$, output~$C_\phi := \{ u \}$, and \KwStop.
    \label{line:TphiTwoBags}
}

Pick a leaf of~$T_\phi$ and make it the root of~$T_\phi$.
\label{line:rootTphi}

Split $T_\phi$ into a set~$\bbL$ of leaves, a set~$\bbP$ of path segments, and a set~$\bbB$ of branching bags.
\label{line:splitTphi}

Create an empty set~$C_\phi$.

\ForEach
{%
    \( P \in \bbP \)
}
{%
    Find a shortest path~$Q_P$ from $\Sup(P)$ to~$\Sdown(P)$ and add its vertices into~$C_\phi$.
    \label{line:pathSegFindQ}

    Let $x \in \Sup(P)$ be the start vertex and $y \in \Sdown(P)$ be the end vertex of~$Q_P$.
    Set $\nu \big( \Sup(P) \big) := x$ and $\nu \big( \Sdown(P) \big) := y$.
    \label{line:pathSegSep}
}

\ForEach
{%
    \( B \in \bbB \)%
    \label{line:loopBranchBag}
}
{%
    If $\nu \big( \Sup(B) \big)$ is defined, let $u := \nu \big( \Sup(B) \big)$.
    Otherwise, let $u$ be an arbitrary vertex in~$\Sup(B)$ and set $\nu \big( \Sup(B) \big) := u$.

    \parbox[t]{\hsize}
    {%
        \rHrt Let $v := c(B)$ be the center of~$B$. \\
        \rDmd Let $v := u$.
    }
    \label{line:defineV}

    Find a shortest path from $u$ to~$v$ and add its vertices into~$C_\phi$.

    \ForEach
    {%
        \( S_i \in \calSdown(B) \)
    }
    {%
        If $\nu(S_i)$ is defined, let $w_i := \nu(S_i)$.
        Otherwise, let $w_i$ be an arbitrary vertex in~$S_i$ and set $\nu(S_i) := w_i$.

        Find a shortest path from $w_i$ to~$v$ and add the vertices of this path into~$C_\phi$.
        \label{line:conBagSepar}
    }
}

Output $C_\phi$.
\end{algorithm}

\begin{theorem}
Algorithm~\ref{algo:conPhiDomination} computes a connected \( \big( r + (\phi + \lambda) \big) \)-dominating set~\( C_\phi \) with \( |C_\phi| \leq |D_r| \) in \( \calO(nm) \) time.
\end{theorem}

\begin{proof}
Since Algorithm~\ref{algo:conPhiDomination} constructs a set~$C_\phi$ as described above, its correctness follows from Corollary~\ref{cor:CphiConDomSet}.
It remains to show that the algorithm runs in $\calO(nm)$ time.

Computing $T_\phi$ (line~\ref{line:compTphi}) can be done in $\calO(nm)$ time (see Lemma~\ref{lem:computeTr}).
Picking a vertex~$u$ in the case when $T_\phi$ contains at most two bags (line~\ref{line:ifTphiOneBag} to line~\ref{line:TphiTwoBags}) can be easily done in $\calO(n)$ time.
Recall that $T_\phi$ has at most $n$~bags.
Thus, splitting $T_\phi$ in the sets $\bbL$, $\bbP$, and~$\bbB$ can be done in $\calO(n)$ time.

Determining all up-separators in~$T_\phi$ can be done in $\calO(M)$ time as follows.
Process all bags of~$T_\phi$ in an order such that a bag is processed before its descendants, \eg, use a preorder or BFS-order.
Whenever a bag~$B$ is processed, determine a set~$S \subseteq B$ of flagged vertices, store $S$ as up-separator of~$B$, and, afterwards, flag all vertices in~$B$.
Clearly, $S$~is empty for the root.
Because a bag~$B$ is processed before its descendants, all flagged vertices in~$B$ also belong to its parent.
Thus, by properties of tree-decompositions, these vertices are exactly the vertices in~$\Sup(B)$.
Clearly, processing a single bag~$B$ takes at most $\calO(|B|)$ time.
Thus, processing all bags takes at most $\calO(M)$ time.
Note that it is not necessary to determine the down-separators of a (branching) bag.
They can easily be accessed via the children of a bag.

Processing a single path segment (line~\ref{line:pathSegFindQ} and line~\ref{line:pathSegSep}) can be easily done in $\calO(m)$ time.
Processing a branching bag~$B$ (line~\ref{line:loopBranchBag} to line~\ref{line:conBagSepar}) can be implemented to run in $\calO(m)$ time by, first, determining $\nu(S)$ for each separator~$S$ of~$B$ and, second, running a BFS starting at~$v$ (defined in line~\ref{line:defineV}) to connect $v$ with each vertex~$\nu(S)$.
Because $T_\phi$ has at most $n$~bags, it takes at most $\calO(nm)$ time to process all path segments and branching bags of~$T_\phi$.

Therefore, Algorithm~\ref{algo:conPhiDomination} runs in \( \calO(nm) \) total time.
\qed
\end{proof}

\section{Implications for the $p$-Center Problem}

The \emph{(Connected) \( p \)-Center} problem asks, given a graph~$G$ and some integer~$p$, for a (connected) vertex set~$S$ with $|S| \leq p$ such that $S$ has minimum eccentricity, \ie, there is no (connected) set~$S'$ with $\ecc_G(S') < \ecc_G(S)$.
It is known (see, \eg,~\cite{BranChepDrag1998}) that the $p$-Center problem and $r$-Domination problem are closely related.
Indeed, one can solve each of these problems by solving the other problem a logarithmic number of times.
Lemma~\ref{lem:rDomToPCenterApprox} below generalises this observation.
Informally, it states that we are able to find a $+ \phi$-approximation for the $p$-Center problem if we can find a good $(r + \phi)$-dominating set.

\begin{lemma}
    \label{lem:rDomToPCenterApprox}
For a given graph~\( G \), let \( D_r \) be an optimal (connected) \( r \)-dominating set and \( C_p \) be an optimal (connected) \( p \)-center.
If, for some non-negative integer~\( \phi \), there is an algorithm to compute a (connected) \( (r + \phi) \)-dominating set~\( D \) with \( |D| \leq |D_r| \) in \( \calO \big( T(G) \big) \) time, then there is an algorithm to compute a (connected) \( p \)-center~\( C \) with \( \ecc_G(C) \leq \ecc_G(C_p) + \phi \) in \( \calO \big( T(G) \log n \big) \) time.
\end{lemma}

\begin{proof}
Let $\calA$ be an algorithm which computes a (connected) $(r + \phi)$ dominating set~$D = \calA(G, r)$ for~$G$ with $|D| \leq |D_r|$ in $\calO \big( T(G) \big)$ time.
Then we can compute a (connected) $p$-center for~$G$ as follows.
Make a binary search over the integers~$i \in [0, n]$.
In each iteration, set $r_i(u) = i$ for each vertex~$u$ of~$G$ and compute the set~$D_i = \calA(G, r_i)$.
Then, increase~$i$ if $|D_i| > p$ and decrease~$i$ otherwise.
Note that, by construction, $\ecc_G(D_i) \leq i + \phi$.
Let $D$ be the resulting set, \ie, out of all computed sets~$D_i$, $D$ is the set with minimal~$i$ for which~$|D_i| \leq p$.
It is easy to see that finding~$D$ requires at most $\calO \big( T(G) \log n \big)$ time.

Clearly, $C_p$ is a (connected) $r$-dominating set for~$G$ when setting $r(u) = \ecc_G(C_p)$ for each vertex~$u$ of~$G$.
Thus, for each~$i \geq \ecc_G(C_p)$, $|D_i| \leq |C_p| \leq p$ and, hence, the binary search decreases~$i$ for next iteration.
Therefore, there is an $i \leq \ecc_G(C_p)$ such that $D = D_i$.
Hence, $|D| \leq |C_p|$ and $\ecc_G(D) \leq \ecc_G(C_p) + \phi$.
\qed
\end{proof}

From Lemma~\ref{lem:rDomToPCenterApprox}, the results in Table~\ref{tbl:pCenterResults} and Table~\ref{tbl:ConPCenterResults} follow immediately.

\begin{table}
    [htb]
    \centering
    \caption
    {%
        Implications of our results for the $p$-Center problem.
    }
    \label{tbl:pCenterResults}
    \def\arraystretch{1.25}
    \begin{tabular}{lcc}
        \hline
        Approach
            & Approx.
            & Time \\
        \hline
        Layering Partition
            & $+ \Delta$
            & $\calO(m \log n)$ \\
        Tree-Decomposition
            & $+ \rho$
            & $\calO(nm \log n)$ \\
        \hline
    \end{tabular}
\end{table}

\begin{table}
    [htb]
    \centering
    \caption
    {%
        Implications of our results for the Connected $p$-Center problem.
    }
    \label{tbl:ConPCenterResults}
    \def\arraystretch{1.25}
    \begin{tabular}{lcc}
        \hline
        Approach
            & Approx.
            & Time \\
        \hline
        Layering Partition
            & $+ 2 \Delta$
            & $\calO(m \, \alpha(n) \log \Delta \log n)$ \\
        Tree-Decomposition
            & $+ \min (5 \rho, 3 \lambda)$
            & $\calO(nm \log n)$ \\
        \hline
    \end{tabular}
\end{table}

In what follows, we show that, when using a layering partition, we can achieve the results from Table~\ref{tbl:pCenterResults} and Table~\ref{tbl:ConPCenterResults} without the logarithmic overhead.

\begin{theorem}
    \label{theo:pCenterLayPart}
For a given graph~\( G \), a \( + \Delta \)-approximation for the \( p \)-Center problem can be computed in linear time.
\end{theorem}

\begin{proof}
First, create a layering partition~$\calT$ of~$G$.
Second, find an optimal $p$-center~$\calS$ for~$\calT$.
Third, create a set~$S$ by picking an arbitrary vertex of~$G$ from each cluster in~$\calS$.
All three steps can be performed in linear time, including the computation of~$\calS$ (see~\cite{Frederickson1991}).

Let $C$ be an optimal $p$-center for~$G$.
Note that, by Lemma~\ref{lem:LayPartVertDist}, $C$ also induces a $p$-center for~$\calT$.
Therefore, because $S$ induces an optimal $p$-center for~$\calT$, Lemma~\ref{lem:LayPartVertDist} implies that, for each vertex~$u$ of~$G$,
\[
    d_G(u, C) \leq d_G(u, S) \leq d_\calT(u, \calS) + \Delta \leq d_\calT(u, C) + \Delta \leq d_G(u, C) + \Delta.
    \tag*{\qed}
\]
\end{proof}

\begin{theorem}
For a given graph~\( G \), a \( + 2 \Delta \)-approximation for the connected \( p \)-Center problem can be computed in \( \calO \big( m \, \alpha(n) \log \min(\Delta, p) \big) \) time.
\end{theorem}

\begin{proof}
Recall Algorithm~\ref{algo:con2DeltaDom} for computing a connected $(r + 2 \Delta)$-dominating set.
We create Algorithm~\ref{algo:con2DeltaDom}$^*$ by slightly modifying Algorithm~\ref{algo:con2DeltaDom} as follows.
In line~\ref{line:compTr}, instead of computing an $r$-dominating subtree~$T_r$ of~$\calT$, compute an optimal connected $p$-center~$T_p$ of~$\calT$ (see~\cite{YenChen2007}).
Accordingly, in line~\ref{line:compDeltaDom}, compute a $\delta$-dominating subtree of~$T_p$, check in line~\ref{line:checkSdelta} if $|S_\delta| \leq |T_p|$ (\ie, if $|S_\delta| \leq p$), and output in line~\ref{line:outSdelta} the set~$S_\delta$ with the smallest~$\delta$ for which $|S_\delta| \leq p$.

Let $S$ be the set computed by Algorithm~\ref{algo:con2DeltaDom}$^*$.
As shown in the proof of Theorem~\ref{theo:rDeltaConDom}, it follows from Lemma~\ref{lem:SdeltaCardinality} and Corollary~\ref{cor:SdeltaCardinality} that $S$ is connected, $|S| \leq p$, and $S = S_\delta$ for some $\delta \leq \Delta$.

Now, let $C$ be an optimal connected $p$-center for~$G$.
Clearly, by definition of~$C$ and by Lemma~\ref{lem:LayPartVertDist}, $\ecc_G(C) \leq \ecc_G(S_\delta) \leq \ecc_\calT(T_\delta) + \Delta$.
Because $T_\delta$ is a $\delta$-dominating subtree of~$T_p$, $\ecc_\calT(T_\delta) \leq \ecc_\calT(T_p) + \delta$.
Let $T_C$ be the subtree of~$\calT$ induced by~$C$, \ie, the subtree of~$\calT$ induced by the clusters which contain vertices of~$C$.
Then, because $T_p$ is an optimal connected $p$-center for~$\calT$ and, clearly, $|T_C| \leq p$, $\ecc_\calT(T_p) \leq \ecc_\calT(T_C)$.
Therefore, since $\delta \leq \Delta$, $\ecc_G(C) \leq \ecc_G(S_\delta) \leq \ecc_\calT(T_C) + 2 \Delta$ and, by Lemma~\ref{lem:LayPartVertDist}, $\ecc_G(C) \leq \ecc_G(S_\delta) \leq \ecc_G(C) + 2 \Delta$.

As shown in the proof of Theorem~\ref{theo:rDeltaConDom}, the one-sided binary search of Algorithm~\ref{algo:con2DeltaDom}$^*$ has at most $\calO(\log \Delta)$ iterations.
Because $|T_p| \leq p$, $T_p$ contains a cluster with eccentricity at most $\lceil p/2 \rceil$ in~$T_p$.
Therefore, for any $\delta \geq \lceil p/2 \rceil$, $|T_\delta| = |S_\delta| = 1$ and, thus, the algorithm decreases~$\delta$.
Hence, the one-sided binary search of Algorithm~\ref{algo:con2DeltaDom}$^*$ has at most $\calO(\log p)$ iterations.
Therefore, the algorithm runs in at most $\calO \big( m \, \alpha(n) \log \min(\Delta, p) \big)$ total time.
\qed
\end{proof}

\end{document}